\documentclass[a4]{article}

\usepackage{booktabs}   
\usepackage{subcaption} 

\usepackage{tikz-cd}
\usetikzlibrary{shapes.geometric}

\usepackage{pgfplots}
\usepackage{authblk}
\usepackage{amsthm}
\usepackage{amssymb}
\usepackage{mathtools}

\usetikzlibrary{shapes, positioning, calc}

\usepackage{algpseudocode}
\usepackage{algorithmicx}
\usepackage{algorithm}
\usepackage{listings}
\usepackage{varwidth}
\usepackage{bussproofs}
\usepackage{xspace}
\usepackage{xcolor}
\usepackage{fancyvrb}
\usepackage{geometry}
\newtheorem{definition}{Definition}
\newtheorem{lemma}{Lemma}
\newtheorem{theorem}{Theorem}
\newcommand{\souffle}{Souffl\'{e}\xspace}
\newcommand{\DOOP}{\textsc{Doop}\xspace}
\newcommand{\predicate}[1][]{R_{#1}(X_{#1})}
\newcommand{\tuple}[1][]{t_{#1}}
\newcommand{\level}{h}
\newcommand{\constraints}{\psi}

\newcommand{\ptree}{\tau}
\newcommand{\ptrees}{T}

\newcommand{\immediateconsequence}{\Gamma_P}
\newcommand{\program}{P}
\newcommand{\instance}{I}
\newcommand{\inputinstance}{EDB}
\newcommand{\treeconsequence}{\mathcal{T}_P}

\newcommand{\dlrule}{\rho}
\newcommand{\dlimpl}{\textit{ :- }}

\newcommand{\brackets}[1]{\left({#1}\right)}

\newcommand{\braces}[1]{\left \{ {#1} \right \}}

\DeclareMathOperator{\rulenum}{@rule}
\DeclareMathOperator{\levelnum}{@height}

\begin{document}

\title{Provenance for Large-scale Datalog}
\author[$\dagger$]{David Zhao }
\author[$\star$]{Pavle Suboti\'c}
\author[$\dagger$]{Bernhard Scholz}

\affil[$\dagger$]{School of Computer Science, The University of Sydney.}
\affil[$\star$]{School of Computer Science, University College London.}

\maketitle

\begin{abstract}
Logic programming languages such as Datalog have become popular as
Domain Specific Languages (DSLs) for solving large-scale, real-world
problems, in particular, static program analysis 
and network analysis. The logic specifications which model analysis problems, 
process millions of tuples of data and contain hundreds of highly recursive 
rules. As a result, they are notoriously difficult to debug. While the
database community has proposed several data provenance
techniques that address the \emph{Declarative
  Debugging Challenge} for Databases, in the cases of analysis problems, these
state-of-the-art techniques do not scale.

In this paper, we introduce a novel bottom-up Datalog evaluation strategy for debugging:
our provenance evaluation strategy relies on a new provenance lattice that includes
proof annotations, and a new fixed-point semantics for semi-na\"ive evaluation. A debugging query mechanism allows arbitrary provenance queries, constructing partial proof trees of tuples with minimal height.
We integrate our technique into \souffle, a Datalog engine that synthesizes C++ code, and achieve high performance by using specialized parallel data structures. Experiments are conducted with 
\DOOP/DaCapo, producing proof annotations for tens of millions of output tuples. We show that
our method has a runtime overhead of 1.27$\times$ on average while being more flexible than existing state-of-the-art 
techniques.
\end{abstract}

\section{Introduction}
Datalog and other logic specification languages~\cite{so16,flix16,logicblox15,z311} have seen
a rise in popularity in recent years, being widely used to solve
real-world problems including program analysis~\cite{so16,pointsto15},
declarative networking~\cite{eq10,de11}, security analysis~\cite{mv05}
and business applications~\cite{logicblox15}. Logic programming
provides declarative semantics for computation, resulting in succinct
program representations and rapid-prototyping capabilities for scientific and
industrial applications. Rather than prescribing the computational
steps imperatively, logic specifications define the intended result
declaratively and thus can express computations concisely. 
For example, logic programming has gained traction in the
area of program analysis due to its flexibility in building custom
program analyzers~\cite{so16}, points-to  
analyses for Java programs~\cite{doop09}, and security analysis
for smart contracts~\cite{mm18,icse19}.

Despite the numerous advantages, the declarative semantics of Datalog poses
a debugging challenge. Strategies employed in debugging imperative programs
such as inspecting variables at given points in the program execution do not
translate to declarative programming. Logic specifications lack the notions of state 
and state transitions. Instead, they have
relations that contain tuples. These relations may be input relations, such as 
those describing the instance of an analysis, intermediate relations, or 
output relations, such as those containing the results of an analysis. Relations 
can only be viewed in full, without any explanation of the origin or derivation of
data, after the completion of a complicated evaluation strategy. Thus, the 
Datalog user will find the results alone of logic evaluation inconclusive for debugging 
purposes. 

When debugging Datalog specifications, there are two main scenarios: (1) an 
unexpected output tuple appears, or (2) an expected output tuple does not appear. These 
may occur as a result of a fault in the input data, and/or a fault in the logic rules. 
Both scenarios call for mechanisms to explain  how an output tuple is 
derived, or why the tuple cannot be derived from 
the input tuples. The standard mechanism for these explanations is a \emph{proof tree}. In the case of explaining the 
existence of an unexpected tuple, a proof tree describes formally the sequence of rule 
applications involved in generating the tuple. On the other hand, a failed proof tree,
where at least one part of the proof tree doesn't hold,
may explain why an expected tuple cannot be derived in the logic specification. These proof 
trees can be seen as a form of \emph{data provenance} witness, that is, an explanation 
vehicle for the origins of data~\cite{ww01,pd09}. 

In the presence of complex Datalog specifications and large datasets, Datalog debugging 
becomes an even bigger challenge. While recent developments in Datalog 
evaluation engines, such as \souffle~\cite{so16}, have enabled the effective evaluation 
of complex Datalog specifications with large data using
scalable bottom-up evaluation 
strategies~\cite{bb89,td91}, unlike top-down evaluation, bottom-up evaluation 
does not have an explicit notion of a proof tree in its evaluation. Therefore, to facilitate debugging in bottom-up evaluation, 
state-of-the-art~\cite{ec17,sp15,dd12} techniques have been developed that rewrite the 
Datalog specification with provenance information. Using these techniques, 
Datalog users follow a \emph{debugging cycle} which allows them to find 
anomalies in the input relations and/or the logic rules. In such setups, the typical 
debugging cycle comprises the phases of (1) defining an investigation query, 
(2) evaluating the logic specification to produce provenance witness, (3) investigating the 
faults based on the provenance information, and (4) fixing the faults. 
For complex Datalog specifications, the need for re-evaluation for each investigation 
is impractical. For example, \DOOP~\cite{doop09} with a highly precise analysis setting 
may take multiple days to evaluate for medium to large-sized Java programs. Although state-of-the-art approaches scale for the database querying use cases, 
such approaches are not practical for industrial scale static analysis problems.

A further difficulty in developing debugging support for Datalog is providing 
understandable provenance witnesses. Use cases such as program analysis tend to produce  
proof trees of very large height. For example, investigations on medium sized program 
analyses in \DOOP have minimal height proof trees of over 200 nodes. Therefore, 
a careful balance must be struck between enough information and readability in the 
debugging witnesses. 

In this paper, we present a novel debugging approach that targets Datalog 
programs with characteristics of those found in static program analysis. Our approach scales 
to large dataset and ruleset sizes and provides succinct and 
interactively navigable provenance information. 

The first aspect of our technique is a novel Datalog provenance evaluation strategy 
that augments the intensional database (IDB) with \emph{Proof Annotations} and hence 
allows fast proof tree exploration for \emph{all} debugging queries, \emph{without 
the need for re-evaluation}. The exploration uses the proof annotations to construct 
proof trees for tuples \emph{lazily}, i.e., a 
debugging query for a tuple produces the rule and the subproofs of the rule. The 
subproofs when expanded in consecutive debugging queries, will produce
 a \emph{minimal} height proof tree for the given tuple. Our system also supports  
  \emph{non-existence} explanations of a tuple. In this 
case, proof annotations are not helpful since they cannot describe non-existent 
tuples. Thus, we adapt an approach from \cite{ps18} to provide a \emph{user-guided} 
procedure for explaining the non-existence of tuples.

We implement the provenance evaluation strategy in the synthesis framework of 
\souffle~\cite{so16}, to produce \emph{specialized} data structures and an interactive 
debugging query system for each logic specification. Our approach is tightly integrated into the 
\souffle engine, thus achieving high performance and generalizability that no 
previous provenance approach is able to achieve. We demonstrate the feasibility 
of our technique through the complex Java points-to framework, \DOOP, running the 
Java DaCapo benchmark suite, which produces tens of millions of output tuples. We 
demonstrate that the initial implementation of our novel provenance method incurs a 
runtime overhead of 1.27$\times$, and memory consumption overhead of 1.45$\times$ 
on average. Thus, our provenance evaluation strategy is capable of processing large 
datasets with no performance disadvantage compared to existing techniques, while 
being more flexible for answering debugging queries.

Our contributions in this work are as follows:
\begin{itemize}
    \item a provenance evaluation strategy for Datalog specifications, defining a new evaluation  domain based on a provenance lattice which extends the standard Datalog subset lattice with proof annotations,
    \item the leveraging of parallel bottom-up evaluation to give minimal height proof trees, and provenance queries for constructing  minimal height proof trees utilizing proof annotations, allowing effective bug investigation with a minimum number of user interactions,
    \item an efficient and scalable integration of the proof tree generator system into \souffle, using specialized data structures for storing proof annotations, and
    \item large-scale experiments using the \DOOP program analysis framework with DaCapo benchmarks with tens of millions of tuples, measuring on average 1.27$\times$ overheads for runtime and 1.45$\times$ overheads for memory.
\end{itemize}

The paper is organized as follows: In Section~\ref{sec:problem} we motivate our 
provenance method and describe 
its use in a real-world program analysis use case. In 
Section~\ref{sec:approach} we detail the theoretical basis of our method with regards to the provenance evaluation strategy along with the provenance queries to construct proof trees for tuples. We also demonstrate the minimality properties and present the practical solution that results from this theory. In Section~\ref{sec:implementation} we detail the implementation of our system in \souffle. In Section~\ref{sec:experiments} we present experiments that show the feasibility of our provenance system. In Section~\ref{sec:related} we outline related work, and we conclude in Section~\ref{sec:conclusion}.

\section{Motivation and Problem Statement}\label{sec:problem}

\begin{figure*}
\begin{subfigure}[b]{0.23\textwidth}
\begin{Verbatim}
l1: a = new O();
l2: b = a;

l3: c = new P();
l4: d = new P();

l5: c.f = a;
l6: e = d.f;
l7: b = c.f;
l8: a = b;
\end{Verbatim}
\caption{Input Program\label{ex:program}}
\end{subfigure}
\begin{subfigure}[b]{0.2\textwidth}
\begin{Verbatim}
new(a, l1).
assign(b, a).

new(c, l3).
new(d, l4).

store(c, f, a).
load(e, d, f).
load(b, c, f).
assign(a, b).
\end{Verbatim}
        \caption{EDB Tuples\label{ex:relation}}
\end{subfigure}
\begin{subfigure}[b]{0.47\textwidth}
\begin{Verbatim}
r1: vpt(Var, Obj) :- new(Var, Obj).
r2: vpt(Var, Obj) :- assign(Var, Var2),
                     vpt(Var2, Obj).
r3: vpt(Var, Obj) :- load(Var, Y, F),
                     store(P, F, Q),
                     vpt(Q, Obj),
                     alias(P, Y).
r4: alias(Var1, Var2) :- vpt(Var1, Obj),
                         vpt(Var2, Obj),
                         Var1 != Var2.
\end{Verbatim}
\caption{Datalog Points-to Analysis\label{ex:analysis}}
\end{subfigure}
\caption{Program Analysis Datalog Setup~\label{fig:example}}
\end{figure*}

\begin{figure*}
\centering
\begin{tikzpicture}[ el/.style = {inner sep=2pt, align=left, sloped},
every label/.append style = {font=\tiny}
]
    \node[shape=circle,draw=black] (A) at (0,2) {$a$};
    \node[shape=circle,draw=black] (B) at (0,0) {$b$};
    \node[shape=circle,draw=black] (C) at (2,2) {$c$};

    \node[shape=circle,draw=black] (E) at (4,0) {$e$};
    \node[shape=circle,draw=black] (D) at (4,2) {$d$};

    \node[shape=rectangle,draw=black] (L1) at (0,4) {$l_1$};
    \node[shape=rectangle,draw=black] (L2) at (2,4) {$l_3$};
    \node[shape=rectangle,draw=black] (L3) at (4,4) {$l_4$};

    \path [->, dashed](A) edge node[el, above] {$new$} (L1);
    \path [->, dashed, bend left](B) edge node[el, above] {$assign$} (A);
    \path [->, dashed, bend left](A) edge node[el, above] {$assign$} (B);
    \path [->, dashed, bend right](B) edge node[el, below] {$load[f]$} (C);
    \path [->, dashed](A) edge node[el, above] {$store[f]$} (C);
    \path [->, dashed](C) edge node[el, below] {$new$} (L2);
    \path [->, dashed](D) edge node[el, below] {$new$} (L3);
    \path [->, dashed](D) edge node[el, above, rotate=180] {$load[f]$} (E);

\end{tikzpicture}
    \caption{Points-to Input Diagram\label{ex:diagram}}
\end{figure*}
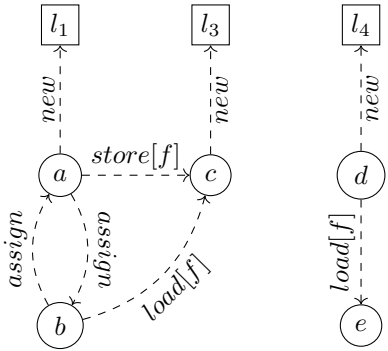

A common approach to characterize the evaluation of a Datalog specification is through \emph{proof trees}. A proof tree for a tuple describes the derivation of that tuple from input tuples and rules. During the debugging cycle, proof trees \emph{explain} why an output tuple exists, and therefore are critical for any investigation into anomalies.
Note that potentially there could be an infinite number of proof trees for the explanation of any given tuple. However, end users 
desire concise proof trees such that that the faulty behaviour of the logic specification is revealed quickly.
In this section, we describe how proof trees can be used to debug a Datalog specification and an overview of our method for generating minimal proof trees for output tuples.

\subsection{Use Case: Program Analysis}

\subsubsection{Points-To Analysis}
We illustrate the utility of debugging via proof trees through a program analysis use case. Figure~\ref{fig:example} illustrates a points-to 
analysis implemented in Datalog. The points-to analysis resembles a field-sensitive but flow-insensitive analysis~\cite{Sridharan:2005:DPA:1094811.1094817}. The input relations (also known as EDB) of the logic specification are the relations \textit{new}, \textit{assign}, \textit{load}, and \textit{store}
express the input program in relational form. The relation \textit{new} represent the object-creation sites of the input program, the
relation \textit{assign}  the assignments, and relations \textit{load}/\textit{store} the read and write accesses of objects via a field. 
Figure~\ref{ex:program} shows an input program encoded in the form of input relations in Figure~\ref{ex:relation}. 
The graph in Figure~\ref{ex:diagram} represents the input relations. The nodes represent either object-creation sites or variables. The edges are object-creation sites, assignments, and load/store instructions.  The graph shows a clear separation between objects of $l_4$ with objects $l_1$ and $l_3$.
The goal is to compute the var-points-to set in the form of the output relation \textit{vpt}. 
The Datalog rules computing the var-points-to set are given in Figure~\ref{ex:analysis}. The first rule makes the variable \texttt{Var} point to object \texttt{Obj}  where \texttt{Obj} is the line number of the object-creation site as an abstraction for all possible objects that could be created by this object-creation site. 
The second rule shows the transfer of the var-points-to set from  source \texttt{Var2} of the assignment to its destination \texttt{Var}. The third rule 
transfers the var-points-to set from the source of a store instruction \texttt{Q} to the destination of a load instruction \texttt{Var}. 
The transfer is conditional 
depending on whether field \texttt{F} of the load and store instructions match and whether the instance variable \texttt{Y} of the load and 
the instance variable \texttt{P} of the store 
instruction alias. The last rule expresses the alias relation between variables \texttt{Var1} and \texttt{Var2}, i.e., two variables alias if they share at least 
one object-creation site \texttt{Obj} in their var-points-to sets. The relations \texttt{vpt} and \texttt{alias} are the output of the analysis and are called the IDB of the Datalog specification.

\subsubsection{Minimal Height Proof Trees}
The analysis example in Figure~\ref{fig:example} computes the output relation $\textit{alias}$ that captures the alias
information of two variables.  A user may investigate \emph{why} a tuple $(a,b)$ exists in the output relation \textit{alias}, 
i.e., how the analysis derives $\textit{alias}(a,b)$ from the input data via the rules. 
Intuitively, this information is contained in the points-to input diagram~(cf.~Figure~\ref{ex:diagram}) 
showing that variables $a$ and $b$ may reach the same object. However, it is not an explanation, as a proof
tree would be, for the tuple $\textit{alias}(a,b)$ as shown in Figure~\ref{fig:ptree}. 
The proof tree shows that $\textit{alias}(a, b)$ is derived by rule $r_4$ using the facts $\textit{vpt}(a, l_1)$ and $\textit{vpt}(b, l_1)$. 
This outcome is expected since it tells us that $a$ and $b$ point to the same object ($l_1$ in this case), and 
thus they may alias.
\begin{figure}[h]
    \begin{prooftree}
        \AxiomC{$new(a, l_1)$}
        \RightLabel{$r_1$}
        \UnaryInfC{$vpt(a, l_1)$}

        \AxiomC{$assign(b, a)$}
        \AxiomC{$new(a, l_1)$}
        \RightLabel{$r_1$}
        \UnaryInfC{$vpt(a, l_1)$}

        \RightLabel{$r_2$}
        \BinaryInfC{$vpt(b, l_1)$}

        \AxiomC{$a \neq b$}

        \RightLabel{$r_4$}
        \TrinaryInfC{$alias(a, b)$}
        
    \end{prooftree}
    \caption{Full proof tree for $alias(a, b)$}\label{fig:ptree}
\end{figure}

The importance of minimality of proof tree height is shown in Figure~\ref{fig:ptree}, which depicts the proof tree resulting from the assignment in line $l_2$ in the input program. In the input program, there is a circular assignment in lines $l_2$ and $l_8$ caused by the flow-insensitivity of the input program, 
and thus the tuple $vpt(b, l_1)$ could be derived in an arbitrary number of rule applications, as shown in Figure~\ref{fig:infiniteptrees}.
\begin{figure}[h]
\begin{prooftree}
    \AxiomC{$assign(a, b)$}

    \AxiomC{$assign(b, a)$}
    \AxiomC{$\ldots$}
    \UnaryInfC{$vpt(b, l_1)$}
    \RightLabel{$r_2$}
    \BinaryInfC{$vpt(a, l_1)$}
    \RightLabel{$r_2$}
    \BinaryInfC{$vpt(b, l_1)$}
\end{prooftree}
    \caption{Infinitely many derivations for $\textit{vpt}(b, l_1)$, resulting from the circular assignment in lines $l_2$ and $l_8$ in the input program \label{fig:infiniteptrees}}
\end{figure}

Thus, even for this small example, there are infinitely many valid proof trees for the tuple $\textit{alias}(a, b)$.
A provenance system ought to produce the most concise proof tree so that an end user can 
understand the derivation of a tuple with the least effort. 

\subsubsection{Proof Tree Fragments for Debugging}
Suppose a Datalog user discovers an unexpected tuple in the output, which indicates that a fault exists somewhere in the 
logic specification. The aim is to investigate the root cause of this fault. Since proof trees provide explanations for the existence of a tuple, 
the proof tree of an unexpected tuple will help identify the fault in the logic specification.

An example fault could be if rule $r_3$ was altered as follows,

\begin{figure}[h]
\centering
\begin{varwidth}{\linewidth}
\begin{verbatim}
r3: vpt(Var, Obj) :- load(Var, Y, F),
                     store(P, F, Q),
                     vpt(Q, Obj),
                     vpt(P, Obj1),
                     vpt(Y, Obj2).
\end{verbatim}
\end{varwidth}
\end{figure}

Note the condition that objects \texttt{P} and \texttt{Y} must alias now no longer holds. 
A minor typo may have introduced this fault, and as a consequence of this typo, the analysis produces the extra tuple $(a, e)$ in 
relation \textit{alias}. This additional tuple becomes a \emph{symptom} of the fault. 
To diagnose this fault, the proof tree of tuple $\textit{alias}(a, e)$ highlights the root cause of the fault.

However, in practice, a full proof tree may be too large to provide a meaningful explanation even if it is of minimal height, and as 
experiments in Section~\ref{sec:experiments} show, proof trees for real-world program analyses (e.g., \DOOP) can exceed heights of 200. Thus a Datalog user may want to explore only relevant \emph{fragments} of it interactively. A fragment of a proof tree is a partial subtree, which consists of some number of levels. For instance, we may construct fragments comprising of 2 levels to explore only parts of the proof tree that are relevant.

We illustrate the exploration of fragments of the proof tree in Figure~\ref{fig:exploration}. In the figure, tuple $\tuple$ denotes the symptom of the fault, i.e., $\tuple$ is an unexpected tuple in the output. The aim is to explore the proof tree for $\tuple$  to find the root cause for this fault. In our example, the user follows the scent of the fault by expanding proof tree fragments that show anomalies. This process produces a path of exploration in the proof tree. The path of exploration discovers the root cause of the fault efficiently, without constructing and displaying the full proof tree of an output tuple.

\begin{figure}
    \tikzset{rounded rectangle/.default={2mm},
          rounded rectangle/.style= { to path={($(\tikztostart)!#1!(\tikztotarget)$)
          -- ($(\tikztostart)!#1!(\tikztotarget)!#1!-90:(\tikztotarget)$)
          -- ($(\tikztostart)!#1!-90:(\tikztotarget)$)
          (\tikztotarget)}}
    }
   \centering
    \begin{tikzpicture}[ el/.style = {inner sep=2pt, align=left, sloped},
    every label/.append style = {font=\tiny}, y=-1cm
    ]
        \node (A) at (2,4) {\tiny $\tuple$};
        \coordinate (B) at (0,0);
        \coordinate (C) at (4,0);

        \coordinate (D) at (1,2);
        \coordinate (E) at (3,2);

        \node (F) at (2.5, 2.3) {\tiny $\tuple[n]$};
        \coordinate (G) at (2, 1);
        \coordinate (H) at (3, 1);

        \node (label) at (1.5, 2.3) {\tiny $\tuple[1]$};
        \node (label2) at (2, 2.3) {\tiny $\ldots$};

        \node (I) at (2.7, 1.3) {};
        \coordinate (J) at (2.2, 0.1);
        \coordinate (K) at (3.2, 0.1);

        \node (x) at (2.9, 0.3) {\tiny $x$};

        \path (A) edge (B);
        \path (B) edge (C);
        \path (A) edge (C);

        \path [dashed] (A) edge (D);
        \path [dashed] (E) edge (D);
        \path [dashed] (E) edge (A);

        \path [dashed] (F) edge (G);
        \path [dashed] (G) edge (H);
        \path [dashed] (H) edge (F);

        \path [dashed] (I) edge (J);
        \path [dashed] (J) edge (K);
        \path [dashed] (K) edge (I);

        \path [red, ->] (A) edge (F);
        \path [red, ->] (F) edge (x);

    \end{tikzpicture}
    \caption{Interactive exploration of fragments of a proof tree for $\tuple$ \label{fig:exploration}}
\end{figure}
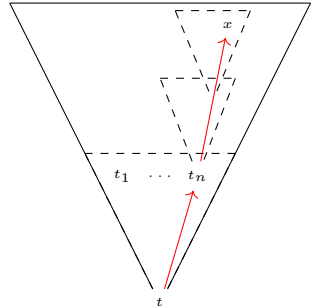

Concretely, we may wish to explain the tuple $\textit{alias}(a, e)$. Figure~\ref{fig:fragment} illustrates the exploration of an explanation for $\textit{alias}(a, e)$ by generating proof tree fragments of 2 levels at a time. The user generates the first fragment and decides that $\textit{vpt}(e, l_1)$ is the most relevant explanation for the fault, and continues down this path. As a result, the root cause (for example, the erroneous rule $r_3$) is discovered after two fragments. This interaction mechanism also justifies the choice to minimize the height of proof trees. By doing this, we minimize the number of user interactions (i.e., proof tree fragments) required to discover the root cause for an anomaly.

Note that provenance queries for databases~(cf.~\cite{sp15, dd12, ec17}) have in general 
shorter evaluation times, and, hence, Database approaches may re-evaluate for each provenance 
query the whole specification. However, for large and complex Datalog specifications, 
provenance queries as shown before are to be performed interactively and promptly, i.e., in a single debugging cycle.
The Datalog user may want to issue multiple provenance queries in a single investigation phase of a debugging cycle. 
Hence, a re-evaluation of the whole logic specification for each provenance query becomes prohibitive, and 
a new provenance approach for Datalog is required.

\begin{figure*}
    \begin{subfigure}[b]{0.9\textwidth}
        \begin{prooftree}
            \AxiomC{$load(e, d, f)$}
            \AxiomC{$store(c, f, a)$}
            \AxiomC{$vpt(a, l_1)$}
            \AxiomC{$vpt(c, l_3)$}
            \AxiomC{$vpt(d, l_4)$}

            \RightLabel{$r_3$}
            \QuinaryInfC{$\color{blue} vpt(e, l_1)$}
        \end{prooftree}
    \end{subfigure}

    \vspace{1em}

    \begin{subfigure}[b]{0.9\textwidth}
        \begin{prooftree}
            \AxiomC{$vpt(a, l_1)$}
            \AxiomC{$\color{blue} vpt(e, l_1)$}
            \AxiomC{$a \neq e$}

            \RightLabel{$r_4$}
            \TrinaryInfC{$\boldsymbol{alias(a, e)}$}
        \end{prooftree}
    \end{subfigure}
    \caption{Exploring the proof of $alias(a, e)$ to find the erroneous rule $r_3$ \label{fig:fragment}}
\end{figure*}

\subsection{Proof Trees and Problem Statement}

We use standard terminology for Datalog, taken from \cite{fd95}. A Datalog specification $\program$ consists of a set of \emph{rules}, of the form $\predicate[0] \dlimpl \predicate[1], \ldots, \predicate[n], \constraints(X_1, \ldots, X_n)$. Each $\predicate[i]$ is a \emph{predicate}, consisting of a \emph{relation name} $R_i$ and an argument $X_i$ consisting of the correct number of variables and constants. The term $\constraints(X_1, \ldots, X_n)$ denotes a conjunction of constraints on the variables in the rule. These constraints may include, for example, arithmetic constraints (such as less than), or negation of a predicate.
A predicate can be \emph{instantiated} to form a \emph{tuple} where each variable is mapped to a constant. An instantiated rule is a rule with each predicate replaced by its instantiation such that the variable mappings are consistent between predicates and the constraints are satisfied. A set of tuples forms an instance $\instance$, and we denote the input instance to be $\inputinstance$.

Given a Datalog specification $\program$, an input instance $\inputinstance$ of $\program$, and a tuple $\tuple$ produced by $\program$, we want to find a \emph{proof tree} of minimal height for $\tuple$. We define a proof tree as follows:
\begin{definition}[Proof Tree]
    Let $\program$ be a Datalog specification, and let  $\inputinstance$ be an input instance. A \emph{proof tree} $\ptree_{\tuple}$ for a tuple $\tuple$ computed by $\program$ is a labeled tree where 
        (1) each vertex is labeled with a tuple,
        (2) each leaf is labeled with an input tuple in $\inputinstance$,
        (3) the root is labeled with $\tuple$, and
        (4) for a vertex labeled with $\tuple[0]$, there is a valid instantiation $\tuple[0] \dlimpl \tuple[1], \ldots, \tuple[n]$ of a rule $\dlrule$ in $\program$ such that the direct children of $\tuple[0]$ are labeled with $\tuple[1], \ldots, \tuple[n]$. Moreover, the vertex is associated with $\dlrule$.
\end{definition}

A proof tree for $\tuple$ can be viewed as an \emph{explanation} for the existence of $\tuple$, by showing how it is derived from other tuples using the rules in the Datalog specification. Importantly, in the context of Datalog, there is a strong connection between a proof tree and the \emph{model} of the specification. A model corresponds to a set of tuples that satisfy a Datalog specification when it is viewed as a constraint system. In a proof tree, each node corresponds to exactly one tuple from the model since each instantiation for the body of a rule generates exactly one tuple.

To formalize the problem statement, we need to characterize proof trees of minimal height. Note that the set of proof trees for a Datalog specification could be constructed inductively by the height of the trees. We denote $\ptree_{\tuple}$ to be a proof tree for tuple $\tuple$, and $\ptrees^{k}$ to be the set of proof trees of height at most $k$. This construction leads to a convenient description of what it means for a proof tree to be of minimal height.

\begin{definition}
    We define the set of all proof trees inductively. Let
    $
        \ptrees^{0} = \braces{\ptree_{\tuple} \mid \tuple \in \inputinstance}
    $
    be the set of proof trees for tuples in the input instance. Then, define $\ptrees^{k}$ in terms of $\ptrees^{k-1}$:
    $
        \ptrees^{k} = \{\ptree_{\tuple} \mid \tuple \dlimpl \tuple[1], \ldots, \tuple[n] $ is a valid instantiation and $ \forall \tuple[i]: \exists \ptree_{\tuple[i]} \in \ptrees^{k-1} \}.
   $
    Then,
    $        \ptrees = \bigcup_{i \geq 0} \ptrees^{i}
    $
    is the set of all proof trees produced by the specification $\program$.
\end{definition}

Note that each $\ptrees^{k}$ consists of proof trees of height at most $k$ since if $\tuple \dlimpl \tuple[1], \ldots, \tuple[n]$ is an instantiation of a rule, then the height of the proof tree for $\tuple$ is equal to the maximum height of the proof trees for $\tuple[1], \ldots, \tuple[n]$ plus 1.
By defining the set of proof trees inductively, a proof tree of minimal height for a given tuple $\tuple$ has height given by 
\begin{align*}
    \min \braces{k \geq 0 \mid \exists \ptree_{\tuple} \in \ptrees^{k}}.
\end{align*}

Intuitively, this means that a proof tree for a tuple $\tuple$ is of minimal height if there does not exist another valid proof tree with a smaller height. We emphasize that a valid proof tree must exist since we have assumed that tuple is in the IDB of the Datalog specification and therefore its existence can be proved.
Based on this inductive construction of proof trees, we reduce the problem of generating a fragment of a proof tree into the following incremental search problem.

\paragraph{Problem statement:} Let $\program$ be a Datalog specification, and $\instance$ be the instance computed by $\program$. Then, given a tuple $\tuple \in \instance$, find the tuples $\tuple[1], \ldots, \tuple[n]$ that form the direct children of $\tuple$ in a minimal height proof tree for $\tuple$, as depicted in Figure~\ref{fig:partialptree}. Denote $\tuple[1], \ldots, \tuple[n]$ to be a \emph{configuration} of the body of the corresponding rule.

\begin{figure}[H]
\centering
    \begin{tikzpicture}[ el/.style = {inner sep=2pt, align=left, sloped},
    every label/.append style = {font=\tiny}, y=-1cm
    ]
        \node (A) at (2,3) {$\tuple$};
        \coordinate (B) at (0,1.8);
        \coordinate (C) at (4,1.8);

        \node (F) at (0.7, 2) {$\tuple[1]$};
        \coordinate (G) at (-0.3, 0);
        \coordinate (H) at (1.7, 0);

        \node (I) at (3.2, 2) {$\tuple[n]$};
        \coordinate (J) at (2.2, 0);
        \coordinate (K) at (4.2, 0);

        \node (dots) at (2, 2) {$\ldots$};

        \path [dashed, red] (A) edge (B);
        \path [dashed, red] (B) edge (C);
        \path [dashed, red] (A) edge (C);

        \path [dashed] (F) edge (G);
        \path [dashed] (G) edge (H);
        \path [dashed] (H) edge (F);

        \path [dashed] (I) edge (J);
        \path [dashed] (J) edge (K);
        \path [dashed] (K) edge (I);

        \path (A.north) edge (F);
        \path (A.north) edge (I);
        \path (A.north) edge (dots);

        \coordinate (top) at (5, 3);
        \coordinate (bot) at (5, 0);

        \path [|-|] (top) edge node[right]{$\min \level$} (bot);

    \end{tikzpicture}
    \caption{One level of a proof tree of minimal height for $\tuple$ \label{fig:partialptree}}
\end{figure}
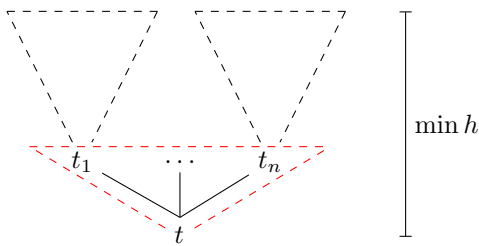

Note that we could recursively construct the subtrees rooted at each $\tuple[i]$, which form valid proof trees for these tuples. Thus, this recursive construction solves the original problem of constructing a fragment of the proof tree of minimal height. Once a certain number of levels have been constructed, or if the only remaining leaves are in the EDB (characterized by having a proof tree consisting of only a single node), then the fragment is complete.

\section{A New Provenance Method}\label{sec:approach}
A simple solution to this problem might be to generate a minimal height proof tree by brute-force searching for matching tuples and return the direct children of the root node. However, this is an unfeasible approach for real-world problems where there may be millions of tuples, and there are also no guarantees that the produced proof trees are of minimal height. Moreover, the two main evaluation strategies for Datalog, \emph{bottom-up} and \emph{top-down} are unsuitable for solving this problem on their own. Bottom-up evaluation is an efficient method for generating tuples but does not store any information related to proof trees. On the other hand, top-down evaluation does compute proof trees as part of its execution, but there are no guarantees for minimality of height. Additionally, to prove the existence of a particular tuple requires proving the existence of every intermediate tuple up to the input tuples, and thus the problem of generating only fragments of proof trees cannot be solved by top-down. Thus, we present a hybrid solution for generating proof trees, consisting of a provenance evaluation strategy based on bottom-up evaluation, plus a debugging query mechanism to construct proof trees.

We summarize the system in Figure~\ref{fig:system}. The Datalog specification and input tuples (EDB) are the input into the system. A pre-processing step of provenance Datalog evaluation generates a set of tuples (IDB) alongside proof annotations for these tuples. The annotated tuples form the input into the interactive proof tree generator system.

The proof tree generator is at the core of the interactive exploration of proofs for tuples. A user queries for a fragment of a proof tree, e.g., \emph{two levels of a proof tree for vpt(b, l1)}, and the system returns the corresponding result. This system can answer any number of queries, and the user can query for \emph{any} fragment of the proof tree for \emph{any} tuple. As previously mentioned, this allows the user to interactively explore the proof for a tuple and find a meaningful explanation for a tuple.

The provenance evaluation strategy resembles a pre-computation step of the Datalog. The evaluation is performed only \emph{once}, but the IDB with proof annotations can subsequently answer \emph{any} debugging query using the same IDB resulting from evaluation. The ability to answer any debugging query without re-evaluation is an advantage over other selective provenance systems~\cite{sp15, ec17}, where the query is given prior to evaluation, which is then instrumented based on the query, and thus evaluation must be performed for each different query.

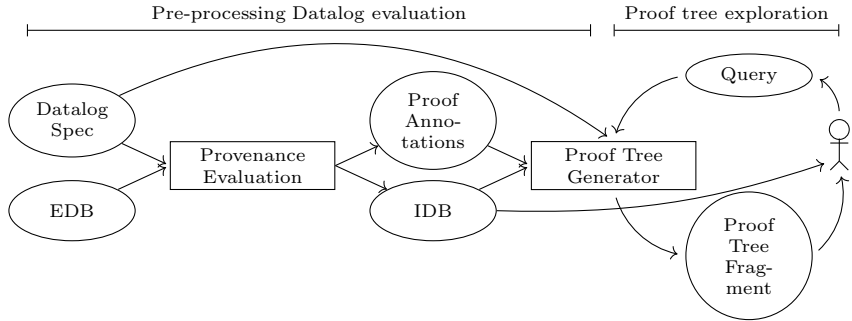
\begin{figure}[ht]
\centering
    \scriptsize
    \begin{tikzpicture}[scale=0.6]
        \usetikzlibrary{shapes}
        \node [draw, ellipse, minimum height=0.8cm, align=center, text width=1cm, name=input] at (0, -1){EDB};
        \node [draw, ellipse, minimum height=0.8cm, align=center, text width=1cm, name=prog] at (0, 1) {Datalog Spec};

        \node [draw, rectangle, align=center, text width=2cm, name=eval] at (4, 0) {Provenance Evaluation};

        \node [draw, ellipse, minimum height=0.8cm,align=center, text width=1cm, name=idb] at (8, -1) {IDB};
        \node [draw, ellipse, minimum height=0.8cm,align=center, text width=1cm,  name=anno] at (8, 1) {Proof Annotations};

        \node [draw, rectangle, align=center, text width=2cm, name=prov] at (12, 0) {Proof Tree Generator};

        \node [draw, ellipse, align=center, text width=1cm, name=query] at (15, 2) {Query};
        \node [draw, ellipse, align=center, text width=1cm, name=tree] at (15, -2) {Proof Tree Fragment};

        \node [draw, circle, minimum size=0.7, name=person_head] at (17, 0.8) {};
        \draw (17, 0.6) -- (17, 0.1);
        \draw (17.2, 0.5) -- (16.8, 0.5);
        \draw (17, 0.1) -- (16.8, -0.1);
        \draw (17, 0.1) -- (17.2, -0.1);

        \node [minimum size=0, name=person_feet] at (17, 0.1) {};

        \path [draw, ->,shorten >= 1.5pt] (input) -- (eval.west);
        \path [draw, ->] (eval.east) -- (anno);
        \path [draw, ->] (eval.east) -- (idb);
        \path [draw, ->,shorten >= 1.5pt] (prog) -- (eval.west);
        \path [draw, ->,shorten >= 3pt] (prog) to [bend left] (prov.north);
        \path [draw, ->,shorten >= 1.5pt] (idb) -- (prov.west);
        \path [draw, ->,shorten >= 1.5pt] (anno) -- (prov.west);

        \path [draw, ->, shorten >= 3pt, shorten <= 3pt] (person_head.north) to [bend right] (query.east);
        \path [draw, ->, shorten >= 3pt, shorten <= 3pt] (query.west) to [bend right] (prov.north);
        \path [draw, ->, shorten >= 3pt, shorten <= 3pt] (prov.south) to [bend right] (tree.west);
        \path [draw, ->, shorten >= 3pt, shorten <= 3pt] (tree.east) to [bend right] (person_feet.south);
        \path [draw, ->, shorten >= 3pt] (idb.east) to [bend right=10] (person_feet.west);

        \draw [|-|] (-1, 3) -- node[above]{Pre-processing Datalog evaluation} (11.5, 3);
        \draw [|-|] (12, 3) -- node[above]{Proof tree exploration} (17, 3);
    \end{tikzpicture}
    \caption{Synthesized Proof Tree Generator system\label{fig:system}}
\end{figure}

\subsection{Standard Bottom-Up Evaluation}
The basis of our approach is the standard bottom-up evaluation strategy for Datalog specification \cite{fd95}. The computational domain of standard bottom-up evaluation is the subset lattice consisting of sets of tuples, denoted \emph{instances} $\instance$. The \emph{na\"ive} algorithm for evaluation is based on the immediate consequence operator, $\immediateconsequence$, which generates new tuples by applying rules in the Datalog specification to tuples in the current instance.
\begin{align*}
    \immediateconsequence\brackets{\instance} = \instance \cup \braces{\tuple \mid \tuple \dlimpl \tuple[1], \ldots, \tuple[n] \text{is a valid instantiation of a rule in } \program \text{ with each } \tuple[i] \in \instance}
\end{align*}

The result of Datalog evaluation is attained when $\immediateconsequence$ reaches a fixpoint, i.e., when $\immediateconsequence\brackets{\instance} = \instance$. Note that this evaluation appears closely related to the inductive construction of proof trees, and indeed the set of tuples represented by $\ptrees^{i}$ is equal to the set of tuples generated by the $i$-th application of $\immediateconsequence$.

However, this na\"ive evaluation will repeat computations since a tuple computed in some iteration will then be recomputed in every subsequent iteration. Therefore, the standard implementation of bottom-up evaluation in real systems such as~\cite{so16,bddbddb05} is \emph{semi-na\"ive}. Semi-na\"ive evaluation contains two main optimizations over na\"ive bottom-up evaluation:
\begin{enumerate}
    \item \textbf{Precedence graph optimization:} the Datalog specification is split into \emph{strata}. Firstly, a precedence graph of relations is computed, then each strongly connected component of the precedence graph forms a stratum. Each stratum is evaluated in a bottom-up fashion as a separate fixpoint computation in order based on the topological order of SCCs. The input to a particular stratum is the output of the previous strata in the precedence graph.
    \item \textbf{New knowledge optimization:} within a single stratum, the evaluation is optimized in each iteration by considering the new tuples generated in the previous iteration. A new tuple is generated in the current iteration only if it directly depends on tuples generated in the previous iteration, therefore avoiding the recomputation of tuples already computed in prior iterations. We describe this process in further detail in Section~\ref{sec:bottomupimplementation}
\end{enumerate}

With these two optimizations, semi-na\"ive performs less repeated computations than the na\"ive algorithm, however, our method for generating proof trees must now be tailored to semi-na\"ive evaluation.

Another essential extension of Datalog is negation, and the standard semantics for negated Datalog is \emph{stratified negation}~\cite{fd95}. A negated predicate is denoted with a $!$ symbol, for example $!vpt(Var, Obj)$ denotes the negation of $vpt(Var, Obj)$. Semantically, a negated predicate evaluates to true if no matching tuples exist in the instance. With stratified negation semantics, a negated predicate is only allowed if the contained variables exist in positive predicates elsewhere in the body of the rule (a condition also known as \emph{groundedness}), and if the corresponding relation does not appear in a cycle in the precedence graph. During evaluation, the stratification of the precedence graph is carried out in a way such that the negated relations can be treated as input into a stratum, and a negated predicate is treated as a constraint, which holds if no corresponding tuple exists in the input instance.

\subsection{Provenance Evaluation Strategy} \label{sec:proveval}
These standard bottom-up evaluation semantics are extended to compute a minimal height proof tree for each tuple. Our extended semantics store \emph{proof annotations} alongside the original tuples. In particular, for each tuple, the annotations are the \emph{height} of the minimal height proof tree, and a number denoting the \emph{rule} which generated the tuple. By using this extra information, we can efficiently generate minimal height proof trees to answer provenance queries (see Section~\ref{sec:guidedtopdown}).

In the context of semi-na\"ive evaluation, and in particular the precedence graph optimization, we describe the provenance evaluation strategy here for a single fixpoint computation (i.e., a single stratum). The resulting correctness properties translate directly to the evaluation of the full Datalog specification since correctness holds for every stratum in the evaluation.

The rule number annotation is easily computed during bottom-up evaluation. With bottom-up evaluation, a new tuple $\tuple$ is generated if there is a rule
$
    \dlrule_k: \predicate \dlimpl \predicate[1], \ldots, \predicate[n], \constraints(X_1, \ldots, X_n)
$
and a set of tuples $\tuple[1], \ldots, \tuple[n]$ such that
$
    \tuple \dlimpl \tuple[1], \ldots, \tuple[n]
$
forms a valid instantiation of the above rule. If this is the case, we say that the rule firing of $\dlrule_k$ generates $\tuple$, and thus the identifier $\dlrule \brackets{\tuple} = k$ is stored as the rule number annotation for $\tuple$. In this way, for each tuple, we track which rule is fired to generate that tuple.

However, the height annotations are more involved and relate closely to the semantics of bottom-up evaluation. Thus, we must develop a formalism for the height annotations, to ensure that it correctly computes the height of the minimal height proof tree for each tuple. To formalize tuples with height annotations, we define a \emph{provenance lattice} as our domain of computation, which extends the standard subset lattice with proof annotations. An element of the provenance lattice is a provenance instance.

\begin{definition}[Provenance Instance]
    A \emph{provenance instance} is an instance of tuples $\instance$ along with a function
    \begin{align*}
        \level : \instance \rightarrow \mathbb{N}
    \end{align*}

    which provides a height annotation of each tuple in the instance. We denote a provenance instance to be the pair $\brackets{\instance, \level}$.
\end{definition}

The aim of these height annotations is to connect a tuple to its proof tree, as depicted in Figure~\ref{tupleandannotation}. The middle value is a tuple along with its height annotation, which is an example of an augmented tuple in a provenance instance. The corresponding proof tree on the right has height matching this annotation.
\begin{figure*}
    \captionsetup[subfigure]{justification=centering}
    \begin{subfigure}[b]{0.14\textwidth}
        \begin{align*}
            vpt(b, l_1)
        \end{align*}
    \end{subfigure}
    \begin{subfigure}[b]{0.05\textwidth}
        \begin{align*}
            \leftrightarrow
        \end{align*}
    \end{subfigure}
    \begin{subfigure}[b]{0.18\textwidth}
        \begin{align*}
            (vpt(b, l_1), 2)
        \end{align*}
    \end{subfigure}
    \begin{subfigure}[b]{0.05\textwidth}
        \begin{align*}
            \leftrightarrow
        \end{align*}
    \end{subfigure}
    \begin{subfigure}[b]{0.38\textwidth}
        \begin{prooftree}
            \AxiomC{$assign(b, a)$}
            \AxiomC{$new(a, l_1)$}
            \RightLabel{$r_1$}
            \UnaryInfC{$vpt(a, l_1)$}

            \RightLabel{$r_2$}
            \BinaryInfC{$vpt(b, l_1)$}
        \end{prooftree}
    \end{subfigure}
    \begin{subfigure}[b]{0.05\textwidth}
        \begin{tikzpicture}
            \coordinate (a) at (0, 0);
            \coordinate (b) at (0, 1.3);

            \draw[|-|] (a) to node[right]{$2$} (b);
        \end{tikzpicture}
    \end{subfigure}
    \caption{Connecting a tuple to a proof tree via a height annotation\label{tupleandannotation}}
\end{figure*}

Similar to the subset lattice of standard bottom-up evaluation, the domain of provenance evaluation should also form a lattice, in our case, based on the subset lattice of standard bottom-up evaluation, but with elements being provenance instances rather than instances. We denote this to be the \emph{provenance lattice} $\mathcal{L}$, where elements are provenance instances. The ordering $\sqsubseteq$ of elements in the lattice is defined by:
\begin{align*}
    \brackets{\instance_1, \level_1} \sqsubseteq \brackets{\instance_2, \level_2} \iff \instance_1 \subseteq \instance_2 \text{ and } \forall \tuple \in \instance_1: \level_1(\tuple) \geq \level_2(\tuple)
\end{align*}

Intuitively, this ordering specifies that an augmented instance $\brackets{\instance_1, \level_1}$ is `less than' another augmented instance $\brackets{\instance_2, \level_2}$ if all tuples in $\instance_1$ also appear in $\instance_2$, with larger or equal height annotation. Therefore, moving `up' the lattice towards the top element results in augmented instances with more tuples and smaller height annotations. This property guarantees the minimality of these height annotations since a bottom-up Datalog evaluation is equivalent to applying a monotone function to move `up' a lattice.

The property that $\sqsubseteq$ a valid partial order is essential to demonstrate that standard properties of Datalog evaluation hold.

\begin{lemma}
    $\sqsubseteq$ is a partial order
\end{lemma}

\begin{proof}
    We show that each of the three properties of a partial order hold.
    \begin{itemize}
        \item Reflexivity: $\brackets{\instance, \level} \sqsubseteq \brackets{\instance, \level}$ since $\instance \subseteq \instance$ and $\forall \tuple \in \instance: \level(\tuple) = \level(\tuple)$

        \item Anti-symmetry: If $\brackets{\instance_1, \level_1} \sqsubseteq \brackets{\instance_2, \level_2}$ and $\brackets{\instance_2, \level_2} \sqsubseteq \brackets{\instance_1, \level_1}$, then $\instance_1 \subseteq \instance_2$ and $\instance_2 \subseteq \instance_1$ (so $\instance_1 = \instance_2$), and $\forall \tuple \in \instance_1: \level_1(\tuple) \geq \level_2(\tuple)$ and $\forall \tuple \in \instance_2: \level_2(\tuple) \geq \level_1(\tuple)$. Since $\instance_1 = \instance_2$, we conclude that $\forall \tuple \in \instance_1: \level_1(\tuple) = \level_2(\tuple)$.

        Thus, $\brackets{\instance_1, \level_1} = \brackets{\instance_2, \level_2}$.

    \item Transitivity: If $\brackets{\instance_1, \level_1} \sqsubseteq \brackets{\instance_2, \level_2}$ and $\brackets{\instance_2, \level_2} \sqsubseteq \brackets{\instance_3, \level_3}$, then we have $\instance_1 \subseteq \instance_2 \subseteq \instance_3$. Also, we have $\forall \tuple \in \instance_1: \level_1(\tuple) \geq \level_2(\tuple)$, and $\forall \tuple \in \instance_2: \level_2(\tuple) \geq \level_3(\tuple)$. Since $\instance_1 \subseteq \instance_2$, we have $\forall \tuple \in \instance_1: \level_1(\tuple) \geq \level_2(\tuple) \geq \level_3(\tuple)$.

        Therefore, we have $\instance_1 \subseteq \instance_3$ and $\forall \tuple \in \instance_1: \level_1(\tuple) \geq \level_3(\tuple)$, and thus $\brackets{\instance_1, \level_1} \sqsubseteq \brackets{\instance_3, \level_3}$.
    \end{itemize}

    Thus, this demonstrates that $\sqsubseteq$ is a partial order.
\end{proof}

In a similar fashion to the immediate consequence operator $\immediateconsequence$ operating on the subset lattice of Datalog instances, provenance evaluation is achieved with a consequence operator $\treeconsequence$ operating on the provenance lattice. The result of evaluation is reached when $\treeconsequence$ reaches a fixpoint, i.e., when $\treeconsequence\brackets{\brackets{\instance, \level}} = \brackets{\instance, \level}$. The main property $\treeconsequence$ is that once a fixpoint has been reached, the proof tree height annotations are \emph{minimal}, and they correspond to the heights of the smallest height proof trees.

The consequence operator $\treeconsequence$ is defined in terms of the $\immediateconsequence$ operator:
\begin{definition}[\emph{Consequence} operator]
    The consequence operator, $\treeconsequence$, generates a new provenance instance:
    \begin{align*}
        \treeconsequence \brackets{\brackets{\instance, \level}} = \brackets{\immediateconsequence \brackets{\instance}, \level'}
    \end{align*}
    where $\level'$ is defined as follows. For any tuple $\tuple \in \immediateconsequence \brackets{\instance}$, let
    \begin{align*}
        G_{\tuple} = \braces{\brackets{\tuple[1], \ldots, \tuple[n]} \mid \tuple \dlimpl \tuple[1], \ldots, \tuple[n] \text{ is a valid rule instantiation with each } \tuple[i] \in \immediateconsequence \brackets{\instance}}
    \end{align*}
    be the set of all configurations of rule bodies generating $\tuple$. Note this may be empty in the case of EDB tuples. Then,
    \begin{align*}
        \level' \brackets{\tuple} = \begin{cases}
            \level \brackets{\tuple} & \text{if } G_{\tuple} = \emptyset \\
            \min_{g \in G_{\tuple}} \braces{\max_{\tuple[i] \in g} \braces{\level\brackets{\tuple[i]}} + 1} & \text{otherwise}
        \end{cases}
    \end{align*}
\end{definition}

The generation of new tuples behaves in the same way as $\immediateconsequence$. To illustrate the height annotations, consider the rule instantiation $vpt(b, l_1) \dlimpl (assign(b, a), 0), (vpt(a, l_1), 1)$, with height annotations written alongside body tuples for convenience. From this rule instantiation, we would generate the tuple $vpt(b, l_1)$ with height annotation $\max \brackets{0, 1} + 1 = 2$. However, the instantiation $vpt(b, l_1) \dlimpl (load(b, c, f), 0), (store(c, f, a), 0), (vpt(a, l_1), 1), (alias(c, c), 2)$ would also generate $vpt(b, l_1)$, but with height annotation $\max \brackets{0, 0, 1, 2} + 1 = 3$. The resulting instance after applying $\treeconsequence$ will contain only the smaller annotation, and thus the resulting provenance tuple is $(vpt(b, l_1), 2)$.

Also, note that this semantics allow for the \emph{update} of the height annotation for a tuple $\tuple \in \instance$. If $\treeconsequence \brackets{\instance, \level} = \brackets{\immediateconsequence \brackets{\instance}, \level'}$ results in $\level'\brackets{\tuple} < \level \brackets{\tuple}$, then we say that the height annotation of $\tuple$ has been updated. An update may happen if $\treeconsequence$ generates new tuples which form a valid configuration of a rule body generating $\tuple$, with lower height annotations than a previous derivation.

We illustrate this definition of provenance evaluation strategy using the running example. We denote $(\tuple, \level)$ to be a tuple $\tuple$ with height annotation $\level$. Assume, for the sake of illustration, that the input instance is computed from a previous fixpoint, and thus has different height annotations, which may be the case with a single fixpoint in semi-na\"ive evaluation. Note that for the full Datalog specification, all input tuples will have a height annotation of 0. Figure \ref{fig:fixpointexample} shows the derived relations under the fixpoint computation with the provenance evaluation strategy. Importantly, in iteration 3, the height annotation for $vpt(b, l_1)$ is updated as a result of a new derivation using \verb|b = c.f; c.f = a;| is computed, with lower height annotation than the previous derivation which used \verb|b = a;|. This demonstrates that annotations of tuples may be updated after they are initially computed, which is essential to maintain minimality.

\begin{figure}[h]
    \textbf{Input}:
    \begin{align*}
        \{&(new(a, l_1), 0), (assign(b, a), 6), (new(c, l_3), 0), (new(d, l_4), 0),\\
        &(store(c, f, a), 0), (load(e, d, f), 0), (load(b, c, f), 0), (assign(a, b), 0)\}
    \end{align*}

    \textbf{Fixpoint iterations}:
    \begin{align*}
        i_0&: \emptyset \\
        i_1&: \braces{(vpt(a, l_1), 1), (vpt(c, l_3), 1), (vpt(d, l_4), 1)} \\
        i_2&: \braces{(vpt(a, l_1), 1), (vpt(c, l_3), 1), (vpt(d, l_4), 1), (vpt(b, l_1), 7)} \\
        i_3&: \braces{(vpt(a, l_1), 1), (vpt(c, l_3), 1), (vpt(d, l_4), 1), (vpt(b, l_1), 3)} \\
        i_4&: \braces{(vpt(a, l_1), 1), (vpt(c, l_3), 1), (vpt(d, l_4), 1), (vpt(b, l_1), 3)}
    \end{align*}
    \caption{IDB relation $vpt$ in each iteration of the fixpoint computation for the example Datalog specification \label{fig:fixpointexample}}
\end{figure}

It remains to be shown that the provenance evaluation strategy is correct, i.e., that $\treeconsequence$ terminates and results in the same set of tuples as $\immediateconsequence$. Additionally, we must show that the height annotations resulting from provenance evaluation strategy is minimal.

\begin{lemma}
    $\treeconsequence$ computes the same tuples as $\immediateconsequence$ at fixpoint, i.e.
    \begin{enumerate}
        \item $\exists k \text{ s.t. } \treeconsequence \brackets{\treeconsequence^k \brackets{\brackets{\instance, \level}}} = \treeconsequence^k \brackets{\instance, \level}$, and
        \item $\treeconsequence^k \brackets{\instance, \level} = \brackets{\immediateconsequence^k \brackets{\instance}, \level^k}$ for some level annotation function $\level^k$
    \end{enumerate}
\end{lemma}

\begin{proof}
    By definition, $\treeconsequence$ generates tuples in the same fashion as $\immediateconsequence$. Since $\immediateconsequence$ always reaches a fixpoint, say after $l$ iterations, i.e., $\immediateconsequence\brackets{\immediateconsequence^l \brackets{\instance}} = \immediateconsequence^l \brackets{\instance}$, we have
    \begin{align*}
        \treeconsequence^l \brackets{\brackets{\instance, \level}} = \brackets{\immediateconsequence^l \brackets{\instance}, \level^l}
    \end{align*}
    Any further applications of $\treeconsequence$ do not change the set of tuples since $\immediateconsequence$ has already reached a fixpoint. Thus, after $l$ iterations, $\treeconsequence$ computes the same tuples as $\immediateconsequence$. 

    If there exists a $k \geq l$ such that $\treeconsequence$ reaches fixpoint after $k$ iterations, then the theorem is proved. Consider applying $\treeconsequence$ to $\brackets{\immediateconsequence^l \brackets{\instance}, \level^l}$. The set of tuples will not change. For any tuple $\tuple \in \immediateconsequence^l \brackets{\instance}$, the height annotation can only decrease as a result of applying $\treeconsequence$ since $\treeconsequence$ takes the minimum height over all rule configurations generating $\tuple$ and $\level^l \brackets{\tuple}$ also must result from such a configuration.

    The height annotation is bounded from below by $0$ since EDB tuples have non-negative annotations, and each subsequently generated tuple has increasing annotation. Therefore, applying $\treeconsequence$ monotonically decreases the height annotation of $\tuple$, which is bounded from below, so eventually, a fixpoint must be reached. Since this holds for all tuples in $\immediateconsequence^l \brackets{\instance}$, $\treeconsequence$ must reach a fixpoint after $k \geq l$ iterations.
\end{proof}

We have established that the provenance evaluation strategy terminates and computes the same set of tuples as standard bottom-up evaluation. It remains to be shown that the proof height annotations are minimal, i.e., that they reflect the real height of the minimal height proof tree for each tuple, and also that they correspond to real proof trees. The property of minimal height annotations is the major result of this section since it demonstrates that our method generates proof trees of minimal height.

\begin{theorem} \label{thm:minimalheightptrees}
    Let $\treeconsequence^k \brackets{\brackets{\instance, \level}} = \brackets{\immediateconsequence^k \brackets{\instance}, \level^{k}}$ be the resulting instance at fixpoint of $\treeconsequence$. Then, for any arbitrary tuple $\tuple \in \immediateconsequence^k \brackets{\instance}$, 
    \begin{enumerate}
        \item \label{thm:minimalheightptrees1} there does not exist any sequence of tuples $\tuple[1], \ldots, \tuple[n]$ such that
            $
                \tuple \dlimpl \tuple[1], \ldots, \tuple[n]
            $
            is a valid instantiation of a rule in $\program$ with each $\tuple[i] \in \immediateconsequence^k \brackets{\instance}$ and $\level \brackets{\tuple} > \max \braces{\level \brackets{\tuple[1]}, \ldots, \level \brackets{\tuple[n]}} + 1$, and
        \item \label{thm:minimalheightptrees2} there is a valid proof tree for $\tuple$ with height $\level^{k} \brackets{\tuple}$
    \end{enumerate}
\end{theorem}

\begin{proof}
    The proof for part \ref{thm:minimalheightptrees1} is by contradiction. Assume such a sequence of tuples $\tuple[1], \ldots, \tuple[n]$ exists. Consider applying $\treeconsequence$ to the instance.
    \begin{align*}
        \treeconsequence \brackets{\immediateconsequence^k \brackets{\instance}, \level^k} = \brackets{\immediateconsequence^k \brackets{\instance}, \level^{k+1}}
    \end{align*}

    with $\level^{k+1} \brackets{\tuple} = \min_{g \in G_{\tuple}} \braces{\max_{\tuple[i] \in g} \braces{\level^k \brackets{\tuple[i]}} + 1}$ by definition of $\treeconsequence$. The set of tuples does not change since we assume that a fixpoint of $\immediateconsequence$ has already been reached.

    Since the sequence $\tuple[1], \ldots, \tuple[n]$ is a valid rule body configuration generating $t$, it is an element of $G_{\tuple}$, and therefore is considered when updating the height annotation of $\tuple$. Since the height annotation resulting from this sequence is lower than $\level^k \brackets{\tuple}$, the update will happen, and thus a fixpoint has not yet been reached.

    Thus, we have a contradiction, and so such a sequence producing a lower height annotation cannot exist.

    The proof for part \ref{thm:minimalheightptrees2} is by induction on the height annotation of $\tuple$. Let $h = \level^k \brackets{\tuple}$ for simplicity.

    If $h = 0$, then $\tuple$ is in the EDB. In this case, the proof tree with a single node corresponding to $\tuple$ is a valid proof tree. Otherwise, for $h > 0$, assume the hypothesis is true for all tuples with height annotation less than $h$. By definition of $\treeconsequence$, there exists a sequence $\tuple \dlimpl \tuple[1], \ldots, \tuple[n]$ such that
    \begin{align*}
        h = \max \brackets{\level^k \brackets{\tuple[1]}, \ldots, \level^k \brackets{\tuple[n]}} + 1
    \end{align*}

    By the assumption, there are valid proof trees for each $\tuple[i]$ of height $\level^k \brackets{\tuple[i]}$. We can generate a proof tree as follows:
    \begin{prooftree}
        \AxiomC{$\ldots$}
        \UnaryInfC{$\tuple[1]$}
        \AxiomC{$\ldots$}
        \AxiomC{$\ldots$}
        \UnaryInfC{$\tuple[n]$}
        \TrinaryInfC{$\tuple$}
    \end{prooftree}
    where each $\ldots$ represents the subtree forming a valid proof tree for each $\tuple[i]$. This resulting proof tree has height
    \begin{align*}
        \max \brackets{\level^k \brackets{\tuple[1]}, \ldots, \level^k \brackets{\tuple[n]}} + 1
    \end{align*}
    which equals $h$. This forms a valid proof tree for $\tuple$ of height $\level^k \brackets{\tuple}$.
\end{proof}

We have shown the correctness and minimal height annotations of the provenance evaluation strategy for a single fixpoint computation. To evaluate a stratified Datalog specification in a semi-na\"ive fashion, each stratum is evaluated as a separate fixpoint using the provenance evaluation strategy. The correctness of the evaluation of a full Datalog specification follows from the correctness of each fixpoint evaluation.

\subsubsection{Complexity of Provenance Evaluation Strategy}
In this section, we discuss the complexity of the provenance evaluation strategy. We characterize this complexity by the number of rule firings during evaluation. With standard bottom-up evaluation, we say that a rule is fired if it generates a new tuple. Therefore, for each tuple generated by the Datalog specification, there is exactly 1 rule firing. However, with the provenance evaluation strategy, a rule is also fired if it results in an update for the height annotation of a tuple. Therefore, we consider the number of updates performed during evaluation of the specification as a characterization of the extra amount of work done by provenance evaluation compared to standard bottom-up evaluation.

\begin{theorem}
    An upper bound for the number of updates performed is $\mathcal{O} \brackets{n \times \max h}$, where $\max h$ denotes the maximum attained height annotation for any tuple during evaluation and $n$ the number of tuples generated by the specification.
\end{theorem}

\begin{proof}
    To prove this, we need to show two things: (1) that it is a true upper bound, and (2) that it is a tight bound.

    To prove (1), consider a tuple $\tuple$ attaining a height annotation of $\max h$. Its annotation may only be updated if there is a valid derivation for $\tuple$ with a lower height. In the worst case, in each update, we reduce the annotation by $1$, and thus we must perform $\max h$ updates to $\tuple$. Considering all tuples produced by the specification, we may update all tuples in this way in the worst case, and therefore, we have $\mathcal{O} \brackets{n \times \max h}$ updates.

    To prove (2), we show an example attaining the upper bound, in Figure~\ref{fig:upper_bound}. In this example, the maximum 
    height annotation is $2k$, and the tuple $reach(a, e)$ will be updated $k$ times as new derivations are computed using nodes in the bottom chain. Furthermore, each tuple $reach(a, x)$ corresponding to nodes $x$ in the `leg' must be updated $\mathcal{O} \brackets{k}$ 
    times as the tuple $reach(a, e)$ is updated. Since there are $k$ nodes in the leg, each of which is updated $\mathcal{O} \brackets{k}$ times, we have in total $\mathcal{O} \brackets{k^2}$ updates, which coincides with the upper bound. Therefore, this upper bound is tight.
\end{proof}

\begin{figure}[h]
\centering
    \begin{tikzpicture}
        \node[draw, circle] (a) at (0, 0) {\tiny $a$};
        \node[draw, circle] (b) at (1, 0) {};
        \node[draw, circle] (c) at (2, 0) {};
        \node[draw, circle] (d) at (3, 0) {};
        \node[draw, circle] (e) at (3, 1) {\tiny $e$};
        \node[draw, circle] (f) at (4, 1) {};
        \node[draw, circle] (g) at (5, 1) {};
        \node[draw, circle] (h) at (6, 1) {};

        \path[draw, ->] (a) to node[below]{\tiny $0$} (b);
        \path[draw, dashed] (b) to (c);
        \path[draw, ->] (c) to node[below]{\tiny $0$} (d);

        \path[draw, ->] (e) to node[above]{\tiny $0$} (f);
        \path[draw, dashed] (f) to (g);
        \path[draw, ->] (g) to node[above]{\tiny $0$} (h);

        \path[draw, ->] (a) to[bend left] node[pos=0.2, above, sloped]{\tiny $2k$} (e);
        \path[draw, ->] (b) to[bend left] node[pos=0.2, above, sloped]{\tiny $2k-2$} (e);
        \path[draw, ->] (c) to[bend left] node[pos=0.2, above, sloped]{\tiny $2$} (e);
        \path[draw, ->] (d) to node[right]{\tiny $0$} (e);

        \draw[decoration={brace, mirror, raise=0.5cm}, decorate] (a.west) to node[pos=0.5, below, yshift=-0.6cm]{$k$} (d.east);
        \draw[decoration={brace, raise=0.5cm}, decorate] (e.west) to node[pos=0.5, above, yshift=0.6cm]{$k$} (h.east);
    \end{tikzpicture}

\begin{varwidth}{\linewidth}
    \begin{verbatim}

reach(a, N) :- reach (a, M), edge(M, N).
    \end{verbatim}
\end{varwidth}
    \caption{Example Datalog specification demonstrating the upper bound is tight. The label on each edge denotes the height annotation of that tuple. \label{fig:upper_bound}}
\end{figure}
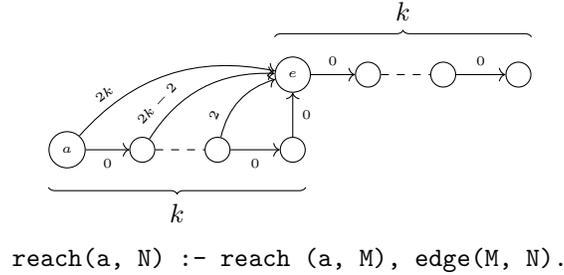

We also note that $\max h$ cannot exceed $n$ since in each iteration of $\treeconsequence$ where a new tuple is generated, the height annotation of that tuple cannot exceed the maximum height annotation in the previous iteration, plus 1. If no more tuples are generated, then the height annotation for any tuple may not increase. Therefore, by generating a new tuple, we increase $\max h$ by at most 1, and therefore this value is at most the total number of tuples generated.

Therefore, in the worst case, the provenance evaluation strategy may have to do a quadratic amount of extra work compared to standard bottom-up evaluation. However, as real-world examples (see Section~\ref{sec:experiments}) show, such instances rarely occur, and scalability is maintained in most real-world cases.

\subsection{Proof Tree Construction by Provenance Queries}\label{sec:guidedtopdown}
Given a provenance instance $\brackets{\instance, \level}$ computed by the provenance evaluation strategy, and a tuple $\tuple \in \instance$, the aim is to construct one level of a minimal height proof tree for $\tuple$. We utilize the height annotations $\level$ and rule number annotations that are stored alongside the instance during bottom-up evaluation. We use a top-down approach for proof tree construction, starting from a query tuple and recursively finding tuples that form a valid instantiation of a rule generating the query tuple. Denote $\level \brackets{\tuple}$ to be the height annotation, and $\dlrule \brackets{\tuple}$ to be the rule corresponding with the rule annotation for $\tuple$.

The result of this search would be a sequence $\tuple[1], \ldots, \tuple[n]$ such that $\tuple \dlimpl \tuple[1], \ldots, \tuple[n]$ is a valid instantiation of $\dlrule \brackets{\tuple}$ leading to a minimal height proof tree. A pre-requisite is that the provenance instance $\brackets{\instance, \level}$ is the result of bottom-up evaluation, and since all possible tuples are computed during this evaluation, we know that each $\tuple[1], \ldots, \tuple[n]$ exists in $\instance$. Thus, this problem would be solved by searching for tuples in the already computed instance $\instance$.

However, we must constrain this search such that the result is part of a proof tree of minimal height since there may be multiple valid configurations for the body of $\dlrule \brackets{\tuple}$, and some configurations may not lead to minimal height proof trees. These constraints result from the annotations from bottom-up evaluation. From Theorem~\ref{thm:minimalheightptrees}, there exists a configuration for the body which leads to a minimal height annotation for the head, and the height annotation for tuple $\tuple$ is generated as
\begin{align*}
    \level \brackets{\tuple} = \max \brackets{\level \brackets{\tuple[1]}, \ldots, \level \brackets{\tuple[n]}} + 1
\end{align*}
by the consequence operator. Therefore, a configuration leading to the minimal height proof tree is $\tuple[1], \ldots, \tuple[n]$ where $\level \brackets{\tuple[i]} < \level \brackets{\tuple}$ for each $\tuple[i]$. Note that there may be multiple configurations leading to a proof tree of minimal height, and any of these configurations is a valid result for the problem.

The problem can be phrased as the following goal search. Given a tuple $\tuple$, and a rule $\dlrule \brackets{\tuple}: \predicate \dlimpl \predicate[1], $ $ \ldots, \predicate[n], \constraints(X_1, \ldots, X_n)$ generating $\tuple$, we want to find tuples $\tuple[1], \ldots, \tuple[n] \in \instance$ such that $\tuple \dlimpl \tuple[1], \ldots, \tuple[n]$ is a valid instantiation of $\dlrule \brackets{\tuple}$, with proof annotations of each $\tuple[i]$ satisfying the former constraints.
\begin{align*}
    ? \dlimpl \predicate[1], \ldots, \predicate[n], \constraints(X_1, \ldots, X_n), \textit{matches}(\tuple, X_1, \ldots, X_n), \level(\predicate[1]) < \level(\tuple) , \ldots, \level(\predicate[n]) < \level(\tuple)
\end{align*}

The condition $matches(\tuple, X_1, \ldots, X_n)$ denotes that for a result $\tuple[1], \ldots, \tuple[n]$, the variable mapping from each $X_i$ to $\tuple[i]$ is consistent with the variable mapping from $X$ to $\tuple$. This is related to the problem of unification in Prolog, and in our context is crucial to ensure that the resulting configuration forms a valid instantiation of $\dlrule$.

\textbf{Example:} We illustrate this construction using the running example. The query is for the tuple $alias(a, b)$. From the initial bottom-up evaluation, the height annotation is $\level\brackets{alias(a, b)} = 4$, and the generating rule is $r_4: alias(Var1, Var2) \dlimpl vpt(Var1, Obj), vpt(Var2, Obj)$.

The search is for tuples forming a configuration for the body of $r_4$, $vpt(Var1, Obj), vpt(Var2, Obj)$ satisfying the constraints
\begin{align*}
    Var1 = a, \\
    Var2 = b, \\
    \level\brackets{vpt(Var1, Obj)} < 4, \\
    \level\brackets{vpt(Var2, Obj)} < 4
\end{align*}

In this example, the first two constraints corresponds with $\textit{matches}(\tuple, X_1, \ldots, X_n)$, and the last two constraints enforce the conditions for proof height annotations. Therefore, the goal search is
\begin{align*}
    ? \dlimpl &vpt(Var1, Obj), vpt(Var2, Obj), Var1 \neq Var2, Var1 = a, Var2 = b,\\
    &\level\brackets{vpt(Var1, Obj)} < 4, \level\brackets{vpt(Var2, Obj)} < 4
\end{align*}

In this case, we find the tuples $vpt(a, l_1), vpt(b, l_1)$, which form the next level of the proof tree:
\begin{prooftree}
    \AxiomC{$vpt(a, l_1)$}
    \AxiomC{$vpt(b, l_1)$}
    \AxiomC{$a \neq b$}
    \RightLabel{$r_4$}
    \TrinaryInfC{$alias(a, b)$}
\end{prooftree}

This goal search can be seen as a special case of Datalog, denoted \textbf{0-IDB} Datalog. The properties of 0-IDB Datalog are:
\begin{itemize}
    \item There are no IDB relations, and every possible tuple is in EDB
    \item There is a single goal search query
    \item The evaluation terminates as soon as the first solution for the goal search is found
\end{itemize}
Our method for searching for body tuples fits into this framework. The IDB resulting from the provenance evaluation strategy 
forms the EDB of our 0-IDB program, and the search for matching body tuples is the single goal search query.

The property that the evaluation terminates as soon as the first solution is found is desirable for efficiency reasons. Our goal search 
contains constraints to ensure that \emph{any} result forms a valid configuration for a minimal height generator, and therefore the 
first solution found is sufficient to solve our problem.

The complexity of the goal search depends highly on the data structures used in the implementation. We assume a
fully (B-Tree) indexed nested loop joins. Therefore searching for a tuple for a rule with an $m$ nested join, 
requires $\mathcal{O}(\log^m n)$ time. Given a proof tree
height of $k$, we need $\mathcal{O}(k \log^m n) \equiv \mathcal{O}(\log^m n)$ to traverse a single branch.

\subsection{Provenance for Non-Existence of Tuples via User Interaction}
The provenance evaluation strategy of the previous section explains the existence of tuples 
in relations. However, the non-existence of tuples may also indicate faults in either the 
input relations and/or in the rules. 

Therefore, we extend our approach explaining on why a tuple \emph{cannot} be 
derived, i.e., if the user expects a tuple, but it does not appear in the IDB,  
the user may wish to investigate why the tuple is not produced. Alternatively, a user may want to understand
why a negated body literal holds in a rule during the debugging process.

A non-existent tuple is characterized by \emph{every} proof tree for the tuple failing to be constructed. The source of failure may be 
 (1) tuples for the construction not being in the EDB/IDB, and/or (2) the constraints of rules not being satisfied. Given the potentially 
 infinite number of failing proof trees, we avoid automatic procedures that 
represent a serious technical challenge and are not guaranteed
to discover a failed proof tree containing the root cause of the fault. In practice,
without a formal description of the root cause of the fault, the 
provenance system cannot decide which failed proof tree is most suitable\footnote{Proof annotations such as introduced in the previous section can only describe existent tuples in the IDB. It is impossible to consider such annotations for tuples that are \emph{not} produced by the specification.}. 

Hence, in our system, we take a pragmatic, semi-automated approach which is 
inspired by existing work such as \cite{ec17,ps18}. Our system leverages user domain 
knowledge and allows user interactions to incrementally guide the construction of a 
single failing proof tree. Each user interaction produces a failing \emph{subproof}, or one level of the proof tree. This failing proof tree provides a succinct representation of valuable 
information for a Datalog user to discover on why an expected tuple is not being produced 
by the specification and does not burden the user with too much unnecessary information.

Formally, we define the problem as follows: given a provenance instance 
$\brackets{\instance, \level}$ computed by the provenance evaluation strategy, a 
tuple $\tuple \notin \instance$, and a rule $\dlrule: \predicate \dlimpl \predicate[1], \ldots, \predicate[n], \constraints(X_1, \ldots, X_n)$ with head relation 
matching $\tuple$, we aim to find a configuration $\tuple[1], \ldots, \tuple[n]$ for the body 
of $\dlrule$, such that either: (1) at least one $\tuple[i] \notin \instance$ or, (2) the 
constraints $\constraints(X_1, \ldots, X_n)$ are not satisfied. Such a configuration forms a failing subproof, and recursively constructing subproofs results in a full failed proof tree. Note that it would be 
impossible to find a configuration where all tuples $\tuple[i] \in \instance$ and 
constraints $\constraints(X_1, \ldots, X_n)$ hold since the prior assumption is that 
$\tuple \notin \instance$. If such an instantiation cannot be found, then the tuple $\tuple$ can be generated by the Datalog specification, and thus $\tuple \in \instance$.

For showing the non-existence of a tuple, the provenance system supports the Datalog user
in constructing the failing proof tree in stages. The debugging query for non-existence has 
three user interaction steps 
that are repeated until the root cause of the fault is found. The user interaction steps are 
as follows:
\begin{enumerate}
\item the user defines a query for the non-existence of a tuple,
\item the user selects a candidate rule from which the tuple may have been derived,
\item the user selects candidate variable values of unbound variables in the rule. 
\end{enumerate}

The system displays the rule application in the failing proof tree 
indicating the portions of the rule that fail 
(i.e., at least one literal / constraint must fail) and the portions of the rule that hold. 

The Datalog user can continue the query with the newly
found failing literals guiding the system to find the root cause of the fault. 
This process is semi-automated since the nature of the fault is known by the Datalog user 
only.

\textbf{Example:} Consider the example from Figure~\ref{fig:example} for which we want to query the non-existence of the 
tuple $vpt(b, l_4)$. In the first user interaction step, the Datalog user queries for an explanation for the non-existence of the tuple  $vpt(b, l_4)$. 
Then, the Datalog user selects an appropriate rule such as rule $r_2$. 

The system can then produce a partial instantiation for the body of the rule, where variables matching the 
head are replaced by concrete values from $\tuple$ such as,
\begin{align*}
    vpt(b, l_4) \dlimpl assign(b, Var_2), vpt(Var_2, l_4)
\end{align*}

In the last user interaction step, the Datalog user selects instantiations for the remaining free variables in rule $r_2$. 
For example, the Datalog user may choose the value $d$ for the free variable $Var_2$.
\begin{align*}
    vpt(b, l_4) \dlimpl assign(b, d), vpt(d, l_4)
\end{align*}
Given the instantiated rule, the provenance system will evaluate which portions of the subproof fail and which portions hold. 
With that information, the Datalog user can continue the exploration of the failing portions to find the root cause 
of the fault. A simple colour labelling helps the Datalog user to indicate which portions fail and hold, respectively.
\begin{prooftree}
    \AxiomC{$\textcolor{red}{\textit{assign}(b, d)\ \text{\sffamily X}}$}

    \AxiomC{$\textcolor{blue}{\textit{vpt}(d, l_4)\ \checkmark}$}
    \RightLabel{$r_2$}
    \BinaryInfC{$\textcolor{red}{\textit{vpt}(b, l_4)}$}
\end{prooftree}
In the above example, the red color and $\text{\sffamily X}$ denotes the non-existence of the tuple $\textit{assign}(b, d)$ in the IDB, i.e., a failing portion of the proof tree. The blue color with $\checkmark$ indicates that 
$\textit{vpt}(d, l_4)$ holds.

In summary, our provenance system constructs a single failed subproof to explain the 
non-existence of a tuple. The construction of the failed subproof is guided by the 
Datalog user to ensure the answer contains a relevant explanation, given the infinitely many possible failed proof trees.  The semi-automatic proof construction approach supports the Datalog 
user by highlighting which portions of the subproof hold and fail, respectively to guide the 
exploration.

\subsection{Alternative Proof Tree Shapes}
Our debugging strategy introduces an \emph{interactive} system to explore fragments of 
proof trees to pinpoint faults in the Datalog specification. Therefore, we 
wish to minimize the number of user interactions required to find the fault. For this aim, 
minimal height proof trees are critical for reducing the number of user interactions in the fault investigation 
phase. The utility of this approach is backed by several user experiences in
industrial-scale applications (see cf. Section 7.1.2 \cite{vldb18}).

While generating proof trees of minimal height is useful for users, in principle our framework 
is more general and can support a variety of metrics that may be beneficial in 
future applications.
In this section, we outline
general properties of proof tree metrics by having the following properties for function $\level$:
\begin{enumerate}
    \item The codomain of $\level$ must have a partial ordering $\sqsubseteq$, so that an update mechanism can be well defined. It is important that the annotation for a tuple can be updated if the same tuple is generated again with smaller (according to $\sqsubseteq$) annotation. This ensures that the resulting annotations are always minimal since tuples will continue being updated with smaller annotations until a fixpoint with annotations is reached.
    \item The metric must be \emph{compositional}, i.e., if $\tuple$ is generated by a rule instantiation $\tuple \dlimpl \tuple[1], \ldots, \tuple[n]$, then $\level \brackets{\tuple} = f \brackets{\level \brackets{\tuple[1]}, \ldots, \brackets{\tuple[n]}}$. The importance of this property is two-fold. Firstly, it ensures that the values of the annotations can be computed during evaluation of the Datalog specification, by encoding $f$ as a functor in the transformed Datalog specification. For example, a rule may be transformed to be $R(X, f \brackets{h_1, \ldots, h_n}) \dlimpl R_1(X_1, h_1), \ldots, R_n(X_n, h_n).$ to compute the value of the annotation.

        Secondly, the compositional property is important for the reconstruction of the 
        proof tree. In the backwards search for a body configuration that may produce 
        the head tuple, $f$ may be encoded as a constraint. For example, a backwards search 
        may be
        \begin{align*}
            ? \dlimpl \predicate[1], \ldots, \predicate[n], \constraints(X_1, \ldots, X_n), \textit{matches}(\tuple, X_1, \ldots, X_n), \level(\tuple) = f \brackets{\level(\predicate[1]), \ldots, \level(\predicate[n])}
        \end{align*}
        where the last constraint ensures that the tuples found from the search correctly generate $\tuple$ with matching annotations.
    \item The metric must be \emph{monotone} and bounded, i.e., $\level \brackets{\tuple[i]} \sqsubseteq \level \brackets{\tuple}$ for all $1 \leq i \leq n$, and \emph{bounded}, i.e., there is a minimum value $c$ such that $c \sqsubseteq \level \brackets{\tuple}$ for any tuple $\tuple$. This property ensures that the provenance evaluation strategy terminates. Monotonicity ensures that with each rule application, the annotation converges towards the minimum value $c$, and once it reaches $c$, then termination must occur.
\end{enumerate}

If a given metric satisfies the above properties, then it can be used instead of proof 
tree height in our framework. Examples of such metrics could be the size of proof trees by 
number of nodes, or a sequence of $k$ proof tree heights describing the smallest $k$ 
proof trees for each tuple. Therefore, our approach could be applicable to wider 
applications than debugging. We leave the integration of other metrics in our provenance evaluation strategy as future work.

\section{Implementation in \souffle} \label{sec:implementation}
In this section, we describe the implementation of our provenance system in 
\souffle~\cite{so16}.  \souffle~\cite{souffle} is an open-source system that is available under the UPL license and is implemented in C++.  
\souffle is a parallel Datalog engine designed for shared memory, multi-core machines, synthesizing highly 
performant parallel C++ code from Datalog specifications. 

Our provenance evaluation strategy and 
proof tree construction system are tightly integrated into the \souffle engine\footnote{Our provenance evaluation strategy is not specific to Souffl\'e -- it can be integrated into any bottom-up evaluation Datalog engine.}. 
Through this tight integration in Souff\'e, we are able to achieve high parallel performance for 
the provenance evaluation, and enable a single evaluation phase to answer multiple 
debugging queries. In contrast, previous approaches~\cite{ec17,sp15,dd12} implement 
a Datalog re-writing scheme, and simply evaluate the re-written Datalog in an existing engine.

The Souffl\'e synthesizer performs a series of specialization steps based on Futamura projections~\cite{futamura99},
which synthesize a C++ program with the same semantics as the Datalog specification.
The main specialization step is the compilation of Datalog into an intermediate representation called Relational Algebra Machine (RAM) which has imperative and relational algebra elements to perform simple relational algebra operations to compute fixed-points for semi-na\"ive evaluation. 
The RAM representation of a Datalog specification is in turn compiled into C++ code. 
The resulting C++ code implements a specialized semi-na\"ive algorithm for the rules in the Datalog specification
that have similar performance to a hand-written program~\cite{so16,pointsto15}. 
In the following, we discuss the implementation of semi-na\"ive evaluation~\cite{fd95} in Souffl\'e and discuss subsequently 
how the synthesised semi-na\"ive evaluation is replaced by the provenance evaluation strategy. 
Semi-na\"ive evaluation avoids re-computations of tuples by assuming that new tuples 
can only be deduced from new tuples in the previous iteration.
This is achieved by creating a \emph{new} and \emph{delta} version of each relation. 
The \emph{new} and \emph{delta} version of a relation store the tuples found in the current iteration and the previous iteration, respectively. 
For example, with a rule
$
    \predicate[0] \dlimpl \predicate[1], \ldots, \predicate[n], \constraints(X_1, \ldots, X_n)
$, each relation $R_k$ is transformed to become a set of relations for each iteration $i$:
\begin{align*}
    R^i_k, new^i_{R_k}, \Delta^i_{R_k}
\end{align*}
where $R^i_k$ stores all the tuples for relation $R_k$ which are computed up until iteration $i$, while $\Delta^i_{R_k}$ stores only the tuples in $R_k$ computed in iteration $i$, without any tuples computed in previous iterations. The relation $new^i_{R_k}$ is an intermediate relation used to compute the $\Delta$ relations.
The essential optimization, compared to na\"ive evaluation, is to realize that in iteration $i+1$, a new tuple is only generated if it directly depends on a tuple generated in iteration $i$. If this condition doesn't hold, i.e., if it depends on knowledge generated in prior iterations, then the tuple would also have been generated in a previous iteration, and thus generating it again would be a redundant computation. 
Thus, a tuple is only generated in iteration $i+1$ if it depends on a $\Delta^i$ relation. This constraint is enforced by transforming the original rule to a set of new Datalog rules which perform semi-na\"ive evaluation:
\begin{align*}
    new^{i+1}_{R_0}(X_0) &\dlimpl \Delta^i_{R_1}(X_1), \predicate[2], \ldots, \predicate[n], \constraints(X_1, \ldots, X_n) \\
    \ldots \\
    new^{i+1}_{R_0}(X_0) &\dlimpl \predicate[1], \ldots, \Delta^i_{R_k}(X_k), \ldots, \predicate[n], \constraints(X_1, \ldots, X_n) \\
    \ldots \\
    new^{i+1}_{R_0}(X_0) &\dlimpl \predicate[1], \ldots, \predicate[n-1], \Delta^i_{R_{n}}(X_n), \constraints(X_1, \ldots, X_n) \\
\end{align*}
Thus, $new^{i+1}_{R_0}$ contains tuples of $R_0$ which depend directly on tuples generated in iteration $i$. The relation $\Delta^{i+1}_{R_0}$ is computed as
\begin{align*}
    \Delta^{i+1}_{R_0} = new^{i+1}_{R_0} - R^i_0
\end{align*}
where the relations are viewed as sets of tuples and $-$ denotes set minus. Thus, $\Delta^{i+1}_{R_0}$ contains only tuples generated in iteration $i$, and no tuples generated in previous iterations. The relation $R^{i+1}_0$ denotes all tuples generated in iterations $0,..,i+1$, and is computed as the union
\begin{align*}
    R^{i+1}_0 = R^i_0 \cup \Delta^{i+1}_{R_0}
\end{align*}
With these auxiliary relations, the final result for the relation $R_0$ is the final $R^i_0$ once a fix-point is reached, i.e., the result of the Datalog specification has stabilized. Note that \souffle evaluates the $\Delta^{i+1}_R$ relation by computing the tuples without an explicit set minus operation since 
an existence check determines whether the tuple already exists in $R^i$ relation before it is inserted into $\Delta^{i+1}_R$. 
For example,  Figure~\ref{fig:existencecheck} depicts a snippet of a \souffle RAM program, part of the semi-na\"ive evaluation of the rule $r_2: vpt(Var, Obj) \dlimpl assign(Var, Var2), vpt(Var2, Obj)$. The join is performed via a loop nest iterating over tuples of relations efficiently via indexes~\cite{vldb18}. 
Line 2 computes the new tuples to be added to the $\Delta^{i+1}_{vpt}$ relation, where the \texttt{NOT IN} operation is an existence check to ensure the generated tuple does not already exist in the $vpt^i$ relation.
\begin{figure}[h]
    \footnotesize
\begin{varwidth}{\linewidth}
    \begin{Verbatim}[numbers=left, xleftmargin=2em]
SCAN assign AS t0 
  SEARCH @delta_vpt AS t1 ON INDEX t1.c0=t0.y WHERE (t0.x,t1.c1) NOT IN vpt
    INSERT (t0.x, t1.c1) INTO @new_vpt
\end{Verbatim}
\end{varwidth}
    \caption{Existence check prior to inserting}
    \label{fig:existencecheck}
\end{figure}

\subsection{Implementing Provenance Evaluation Strategy} \label{sec:bottomupimplementation}
The main challenge of integrating the provenance evaluation strategy is to allow the  synthesis to be aware of proof annotations. 
In particular, the semi-na\"ive evaluation machinery must be replaced by the provenance evaluation strategy as described in Section~\ref{sec:proveval} to handle the proof annotations. 
Another critical part of this machinery is the synthesis of data structures~\cite{vldb18,ppopp19} for relations that 
are specialized for the operations in the program. The synthesized data structures have to be
extended for proof annotations as well, enabling an update semantics in Datalog for the annotations. 

For the provenance evaluation strategy, we need to amend relations by extra attributes to contain the proof annotations. 
We utilize the synthesis pipeline of Souffl\'e by introducing two provenance attributes for each relation. The first attribute represents the rule number of the rule which generated the tuple, and the second attribute represents the proof tree height. 
These two new attributes are introduced for each relation at the syntactic level in Souffl\'e. 
A predicate $\predicate$ is transformed into $R(X, \rulenum, \levelnum)$. 
For the sake of readability in this text, we distinguish between original and provenance tuples, where an original tuple is a provenance tuple without proof annotations.
We rewrite all logic rules at the syntactic level to take account of the two provenance attributes constituting the proof annotation for our system, and to compute the value of the annotations. 
The proof annotation instrumentation is performed as follows where a rule
\begin{align*}
    \dlrule_k: \predicate \dlimpl \predicate[1], \ldots, \predicate[n], \constraints(X_1, \ldots, X_n)
\end{align*}
is transformed into:
\begin{align*}
    \dlrule_k: \ &R(X, k, \max(\levelnum_1, \ldots, \levelnum_n) + 1) \dlimpl \\
    &R_1(X_1, \_, \levelnum_1) ,\ldots, R_n(X_n, \_, \levelnum_n), \constraints(X_1, \ldots, X_n)
\end{align*}
The transformed provenance rule computes level and height annotations for a new tuple, according to the semantics in Section~\ref{sec:proveval}.  
Since the rules are known statically,  the rule number annotation $k$ can be assigned a constant value for each rule in the transformed specification.
The rule numbers of the body predicates are ignored by using $\_$ in each body 
predicate since they do not influence the head predicate. 

The transformed provenance rule syntactically represents the computation of proof annotations during rule evaluation. However, the actual execution of provenance rules differs from a standard semi-na\"ive evaluation as presented in~\cite{fd95}.
The reason is the update mechanism: a newly discovered provenance tuple may overwrite an existing provenance tuple if they are the same original tuple, but the new tuple has a smaller height annotation.

The provenance evaluation strategy extends the semi-na\"ive algorithm by updating the rule number and the height annotation of a tuple (as defined by $\treeconsequence$) if the original tuple already exists and the newly generated tuple has smaller height annotation. In other words, if a smaller proof tree could be found in a subsequent iteration for the same tuple, then an update occurs. Otherwise, if the original tuple doesn not already exist, the provenance tuple is inserted into the relation as is. Thus, the rule computing $\Delta^{i+1}_{R_0}$ in the semi-naive evaluation is modified to accommodate the possibility of updates, i.e., 
\begin{align*}
    \Delta^{i+1}_{R_0} = \brackets{new^{i+1}_{R_0} - R^i_0} \cup \braces{\tuple \in R^i_0 \mid h^{i}(\tuple) > h^{i+1}(\tuple)}
\end{align*}
where $h^i$ denotes the height annotations in iteration $i$.

Therefore, with our provenance evaluation strategy, we integrate the possibility of an annotation update into the data structure. During an insertion operation, if the same original tuple is discovered, with a larger height annotation than the current tuple, then an update occurs. Therefore, we wish to call the insert operation if \emph{either} the tuple does not already exist \emph{or} the existing tuple has the larger proof annotation. A specialized existence check is implemented, implicitly implementing the semantics of the provenance $\Delta^{i+1}_R$ relation. Similar to the standard existence check, the special existence check is implemented as part of the data structure that has been specialized for each relation. 

The result is the RAM snippet in Figure~\ref{fig:provloopnest}. The obvious differences compared to Figure~\ref{fig:existencecheck} are in lines 3-4, where the level annotation is computed within the loop nest, as part of the insertion. Furthermore, the \texttt{PROV NOT IN} operation in line 3 denotes the special provenance existence check, which allows the \texttt{INSERT} to proceed if either the tuple does not already exist, or exists with a larger proof height annotation. Thus, this implements the update semantics discussed above.

\begin{figure}[h]
    \footnotesize
\begin{varwidth}{\linewidth}
    \begin{Verbatim}[numbers=left, xleftmargin=2em]
SCAN assign AS t0
  SEARCH @delta_vpt AS t1 ON INDEX t1.x=t0.y
    IF (t0.x,t1.y,_,(max(t0.@height,t1.@height)+number(1))) PROV NOT IN vpt
      INSERT (t0.x, t1.y, number(2), (max(t0.@height,t1.@height)+number(1))) INTO @new_vpt
\end{Verbatim}
\end{varwidth}
    \caption{Provenance version of RAM loop nest}
    \label{fig:provloopnest}
\end{figure}

However, the specializations in the data structures still remain to be discussed. Souffl\'e employs a highly specialized parallel B-Tree data structure~\cite{ppopp19}, with index orderings for the attributes generated automatically via an optimization problem~\cite{vldb18}. The B-Tree employs a special optimistic read/write lock for each node, allowing high throughput for parallel insertion. During an \verb|insert| operation, a thread may obtain a read lease for each node as it checks whether the tuple to be inserted already exists. If an insertion is required, it checks if the lease has changed, and restarts the whole procedure if it has. Otherwise, it upgrades to a write lease and inserts the tuple into the correct position in the B-Tree. The data structure also takes advantage of \souffle's Datalog evaluation setting, where a single relation is either read from, or written to, but never at the same time. Therefore, there are no interleaved reads and writes, and so read operations are not synchronized.

With the proof annotations, we modify these specializations so that they can take into account the provenance semantics. The important step is the \emph{update} semantics, and thus we integrate an update mechanism into the \verb|insert| operation, without requiring to delete and then re-insert.
The provenance evaluation strategy requires two main modifications to our B-Tree data structure. Firstly, the existence check for insertion should consider only the original tuple, and ignore annotations. This ensures that Datalog set semantics are preserved and that no duplicate original tuples can exist. However, note that we still need to retrieve the full tuple, including its proof annotations. This is important for the proof tree construction, discussed in the next section. To address this concern, we use different lexicographical orderings of indices for the \verb|insert| and \verb|retrieve| operations. The \verb|insert| index order does not include the attributes storing the proof annotations (so that annotations are not considered when checking existence), while the \verb|retrieve| index order does. We also need to ensure that updating an annotation does not change the ordering of tuples according to the index; otherwise subsequent index supported searches will fail. Therefore, the \verb|retrieve| index order requires the annotation attributes to be at the end, as this guarantees that an update to the annotations does not affect tuple ordering.

Secondly, we must have a mechanism to update existing tuples to implement the update semantics in Section~\ref{sec:proveval}. To achieve this, we modified the \verb|insert| operation so that it may also update any existing tuple with a smaller proof annotation. This insertion first requires to check if the original tuple exists in the B-Tree. If it does, rather than aborting (as it would with standard Datalog evaluation), the insertion then checks the annotation. If the height annotation of the existing tuple is larger than the tuple to be inserted, then the annotations of the existing annotation are updated with new values. Note that during an update, a read lease also needs to be validated and upgraded to a write lease. The integration of the update into the \verb|insert| operation avoids any need to delete and re-insert tuples, thus improving the efficiency of the provenance evaluation strategy.
These modifications to the B-Tree reflect the desired insertion semantics, and updates are handled directly in the \verb|insert| operation. All other retrieval operations for the B-Tree are not modified, and tuples can be retrieved as normal, including their proof annotations.

\subsection{Implementing a Proof Tree Construction User Interface}
After the provenance evaluation strategy is completed, the proof tree construction stage is driven by the user.
It is critical that this process is also fully parallelized and highly performant.

The user interface is implemented as a command line, where the user can enter queries to explain the existence and non-existence of a tuple. For example, the query \texttt{explain alias("a", "b")} results in the proof tree in Figure~\ref{fig:explainscreenshot}. Explaining the non-existence of a tuple, i.e., the query \texttt{explainnegation vpt("b", "l4")}, results in the interaction in Figure~\ref{fig:negationscreenshot}. The user may also select the size of proof tree fragments to display, i.e., \texttt{setdepth 6} instructs the system to construct 6 levels of the proof tree in the next query. For each debugging query, the system invokes the relevant procedure to construct a proof tree fragment.

\begin{figure}[h]
    \footnotesize
    \centering
\begin{varwidth}{\linewidth}
\begin{verbatim}
> explain alias("a", "b")
                                 new("a", "l1")
                                 -----------(R1)
new("a", "l1")  assign("b", "a") vpt("a", "l1")
-----------(R1) -----------------------------(R2)
vpt("a", "l1")           vpt("b", "l1")          "a" != "b"
---------------------------------------------------------(R1)
                 alias("a", "b")
\end{verbatim}
\end{varwidth}
\caption{Explaining the tuple $alias(a, b)$}
\label{fig:explainscreenshot}
\end{figure}

\begin{figure}[h]
    \footnotesize
    \centering
    \begin{BVerbatim}[commandchars=\\\{\}]
Enter command > explainnegation vpt("b", "l4")
1: vpt(Var,Obj) :-
   new(Var,Obj).

2: vpt(Var,Obj) :-
   assign(Var,Var2),
   vpt(Var2,Obj).

3: vpt(Var,Obj) :-
   load(Var,Y,F),
   store(P,F,Q),
   alias(P,Y),
   vpt(Q,Obj).

Pick a rule number: 2
Pick a value for Var2: d
assign("b", "d") X  vpt("d", "l4") \checkmark
----------------------------------(R2)
            vpt("b","l4")
\end{BVerbatim}
    \caption{Explaining the non-existence of the tuple $vpt("b", "l4")$}
\label{fig:negationscreenshot}
\end{figure}

It is critical that the proof tree construction procedures are highly performant since the constructed IDB may be very large, and we may need to search through many tuples to construct a proof tree fragment. Therefore, the proof tree construction procedures must be tightly integrated into the \souffle system to enable a high-performance, parallel search. We integrate these procedures into the \souffle RAM, utilizing the existing translation from RAM to parallel C++. Moreover, since the provenance evaluation strategy uses specialized B-Tree data structures, the proof tree construction can also utilize index supported searches to find relevant tuples.

Recall that the proof tree construction is facilitated by searches for subproofs. Therefore, we require a specialized framework in the \souffle RAM to implement a subproof search. We term this framework a \emph{subroutine} framework. Each subproof search can be implemented as a subroutine, thus integrating with the \souffle RAM.

To explain the existence of a tuple, a subproof search is required to search the body of a rule for matching body tuples, satisfying the constraint that proof tree height is lower than the current tuple. This backwards search for a single rule is implemented as a subroutine. For example, the rule $r_2: vpt(Var, Obj) \dlimpl assign(Var, Var2), vpt(Var2, Obj)$ is implemented as the subroutine in Figure~\ref{fig:subroutine}.

\begin{figure}[h]
    \footnotesize
\begin{varwidth}{\linewidth}
    \begin{Verbatim}[numbers=left, xleftmargin=2em]
SUBROUTINE vpt_2_subproof
  SCAN assign AS t0 WHERE t0.x = argument(0) AND t0.@level_number < argument(2)
    SEARCH vpt AS t1 ON INDEX t1.x=t0.y AND t1.y=argument(1) WHERE t1.@level_number < argument(2)
      RETURN (t0.x, t0.y, t0.@rule_number, t0.@level_number, t0.y,
              t1.y, t1.@rule_number, t1.@level_number)
\end{Verbatim}
\end{varwidth}
    \caption{Subroutine for example program}
    \label{fig:subroutine}
\end{figure}

Lines 2-5 represent a search through a database that is already constructed by the initial bottom-up evaluation, to find tuples which satisfy the constraints required for the construction of a proof tree fragment. The values of \verb|argument(0)| and \verb|argument(1)| are the values in the head tuple, and \verb|argument(2)| is the height annotation of the head tuple. Therefore, this subproof search is paramterized by the head tuple. The relations of the body atoms, $assign$ and $vpt$ are searched to find tuples \verb|t0| and \verb|t1| which match the body of the rule. Importantly, the constraints for the level number are encoded in lines 3-4, ensuring that the resulting tuples have level number annotations lower than the query tuple. As shown in Section~\ref{sec:guidedtopdown}, being able to apply this operation recursively allows us to generate the full proof tree.

Similarly, to generate a failed subproof to explain the non-existence of a tuple, the search for failing and holding parts of a subproof is implemented as a subroutine. Given an instantiated rule (which is produced via user interaction), a subroutine returns whether each body tuple is in the IDB and whether each constraint is satisfied.

\section{Experiments}\label{sec:experiments}
In this section, we conduct experiments with the provenance evaluation strategy and provenance queries 
implemented in \souffle (see Section~\ref{sec:implementation}). The experiments are conducted 
for large-scale Datalog specifications. We have the following experimental research claims:
\begin{itemize}
    \item[] \textbf{Claim-I: Scalable Provenance Evaluation Strategy.} Our provenance evaluation strategy only has a minor impact on the runtime  performance, i.e., our provenance evaluation strategy remains scalable for realistic datasets and rulesets.
    \item[] \textbf{Claim-II: Scalable Proof Tree Construction.} The provenance queries for exploring proof paths scales to large proof 
     tree fragments, allowing efficient interactive exploration of proof trees.
    \item[] \textbf{Claim-III: Need for Proof Tree Exploration.} For realistic benchmarks, minimal proof trees are still very large substantiating the need for interactive proof tree exploration.
\end{itemize}

As an experimental testbed, we use the \DOOP~\cite{doop09} points-to analysis framework. We experiment with \DOOP's \textit{context-insensitive} and \textit{1-object-sensitive, 1-heap} (1-obj, 1-heap) analyses that exhibit different
runtime complexities.  As inputs for the points-to analyses, we compute the points-to sets for the DaCapo 2006 Java program benchmarks. Each analysis contains approx. 300 relations, 850 rules and produces up to approx. 26 million output tuples on the DaCapo benchmarks  (See Table~\ref{table:doopstats}). Our experiments were performed on a computer with an Intel Xeon Gold 6130 CPU and 192 GB of memory, running Fedora 27. \souffle executables were generated using GCC 7.3.1.

\begin{table}[t]
    \centering
\footnotesize
    \begin{tabular}{|l|r|r|r|r|}
        \hline
        & \multicolumn{2}{c|}{context-insensitive} & \multicolumn{2}{c|}{1-obj, 1-heap} \\
        \cline{2-5}
        \textbf{Benchmark}& \multicolumn{1}{c|}{\# EDB}  & \multicolumn{1}{c|}{\# IDB}  & \multicolumn{1}{c|}{\# EDB} & \multicolumn{1}{c|}{\# IDB} \\
        \hline
        antlr       & 8,319,095 & 21,832,232 & 8,319,095 & 24,145,648 \\
        bloat       & 4,468,277 & 13,104,020 & 4,468,277 & 15,417,516 \\
        chart       & 8,743,770 & 22,975,742 & 8,743,729 & 25,289,200 \\
        eclipse     & 4,389,770 & 13,076,265 & 4,389,799 & 15,389,708 \\
        fop         & 8,769,583 & 22,970,533 & 8,769,572 & 25,283,913 \\
        hsqldb      & 9,007,087 & 24,561,921 & 9,007,087 & 26,875,437 \\
        jython      & 5,203,400 & 17,158,375 & 5,203,400 & 19,471,797 \\
        luindex     & 4,396,394 & 13,415,336 & 4,396,394 & 15,728,788 \\
        lusearch    & 4,396,415 & 13,415,390 & 4,396,394 & 15,728,788 \\
        pmd         & 8,388,202 & 22,853,676 & 8,388,202 & 25,167,134 \\
        xalan       & 8,670,980 & 23,488,951 & 8,670,966 & 25,802,385 \\
        \hline
    \end{tabular}
    \caption{Statistics for \DOOP benchmarks}
    \label{table:doopstats}
\end{table}

\subsection{Performance of the Provenance Evaluation Strategy}
In Table~\ref{table:doopdacapo}, we present the runtime and memory consumption of \souffle with 8 threads, comparing standard \souffle with our provenance evaluation strategy with proof annotations. We use the DaCapo benchmarks with both the context-insensitive and 1-obj, 1-heap analysis.
As expected, \souffle with proof annotations incurs an overhead during evaluation. This overhead for the provenance evaluation strategy is typically within a factor of 1.3 which is a small overhead to pay for being able to generate minimal proof trees for all possible tuples in the IDB. Hence, we demonstrate the viability of the provenance evaluation strategy for large-scale Datalog specifications, substantiating Claim I.  We noticed that the runtime overhead for the context-insensitive analysis was smaller across all benchmarks than that of the 1-obj-1-heap analysis due to cache locality that was more prominent for smaller memory footprints. Note that the overhead for memory consumption is similar to performance overheads, at approximately 1.45$\times$. This overhead results from the storage of extra proof annotations during evaluation. 

In contrast, a na\"ive direct encoding approach (cf. Chapter 5, ~\cite{davidsthesis}), 
where each tuple is annotated with its full subproof (i.e., direct children in the proof tree), resulted in excessive memory usage (up to 100$\times$) on a simple transitive closure experiment with 2000 tuples, and thus cannot be deployed for large-scale Datalog specifications such as those found in \DOOP.
\begin{table}[t]
    \centering
\footnotesize
    \begin{tabular}{| l | r | r | r | r | r | r |}
        \hline
        & \multicolumn{3}{c|}{Runtime (sec)} & \multicolumn{3}{c|}{Memory (MB)} \\
        \cline{2-7}
        \textbf{Benchmark} & No Prov. & Prov. & ($\times$) & No Prov. & Prov. & ($\times$) \\
        \hline
        \multicolumn{7}{|l|}{\textbf{context-insensitive}} \\
        \hline
        antlr       & 9.73 & 12.29 & 1.26 & 595 & 900 & 1.51 \\
        bloat       & 9.54 & 12.25 & 1.28 & 596 & 900 & 1.51 \\
        chart       & 15.89 & 19.60 & 1.23 & 1,103 & 1,604 & 1.45 \\
        eclipse     & 9.64 & 11.76 & 1.22 & 593 & 898 & 1.51 \\
        fop         & 15.57 & 19.48 & 1.25 & 1,079 & 1,579 & 1.46 \\
        hsqldb      & 16.36 & 19.73 & 1.21 & 1,124 & 1,642 & 1.46 \\
        jython      & 11.00 & 13.62 & 1.24 & 731 & 1,090 & 1.49 \\
        luindex     & 9.62 & 12.00 & 1.25 & 594 & 905 & 1.52 \\
        lusearch    & 9.80 & 12.23 & 1.25 & 593 & 904 & 1.52 \\
        pmd         & 15.58 & 18.90 & 1.21 & 1,053 & 1,542 & 1.46 \\
        xalan       & 15.59 & 19.54 & 1.25 & 1,091 & 1,595 & 1.46 \\
        \hline
        geo-mean    & & & 1.24 & & & 1.44 \\
        \hline
        \multicolumn{7}{|l|}{\textbf{1-obj, 1-heap}} \\
        \hline
        antlr       & 10.84 & 12.60 & 1.16 & 936 & 1,310 & 1.40 \\
        bloat       & 15.77 & 22.00 & 1.40 & 732 & 1,082 & 1.48 \\
        chart       & 21.84 & 28.13 & 1.29 & 1,242 & 1,788 & 1.44 \\
        eclipse     & 15.76 & 21.00 & 1.33 & 729 & 1,080 & 1.48 \\
        fop         & 22.21 & 29.63 & 1.33 & 1,216 & 1,756 & 1.44 \\
        hsqldb      & 23.01 & 29.43 & 1.28 & 1,256 & 1,823 & 1.45 \\
        jython      & 17.54 & 22.96 & 1.31 & 868 & 1,270 & 1.46 \\
        luindex     & 15.94 & 21.55 & 1.35 & 730 & 1,086 & 1.49 \\
        lusearch    & 15.95 & 21.25 & 1.33 & 731 & 1,087 & 1.49 \\
        pmd         & 21.58 & 28.21 & 1.31 & 1,190 & 1,725 & 1.45 \\
        xalan       & 22.09 & 28.68 & 1.30 & 1,224 & 1,773 & 1.45 \\
        \hline
        geo-mean    & & & 1.31 & & & 1.46 \\
        \hline
    \end{tabular}
    \caption{Runtime and memory usage overheads for \souffle with and without proof annotations with 8 threads}
    \label{table:doopdacapo}
\end{table}

\begin{figure*}
\captionsetup[subfigure]{justification=centering}
\begin{subfigure}[b]{0.47\textwidth}
    \center
    \begin{tikzpicture}
        \tikzstyle{every node}=[font=\scriptsize],
        \begin{axis}[
                ybar,
            width=\textwidth,
            height=0.7\textwidth,
            bar width=2pt,
            enlargelimits=0.0,
            clip=false,
            enlarge x limits=0.1,
            legend style={
              anchor=north,legend columns=-1},
            ylabel={Runtime (sec)},
            xlabel={Threads},
            symbolic x coords={1, 2, 3, 4, 5, 6, 7, 8, 9, 10, 11, 12, 13, 14, 15, 16},
            xtick=data,
            ymax=1200,
            ymin=0,
            visualization depends on=rawy\as\rawy, 
            after end axis/.code={ 
                \draw [ultra thick, white, decoration={snake, amplitude=1pt}, decorate] (rel axis cs:0,1.05) -- (rel axis cs:1,1.05);
            },
        ]
            \addplot coordinates {(1, 764) (2, 545) (3, 452) (4, 407) (5, 379) (6, 359) (7, 347) (8, 341) (9, 337) (10, 330) (11, 326) (12, 326) (13, 328) (14, 325) (15, 322) (16, 321)};
            \addplot coordinates {(1, 1112) (2, 760) (3, 616) (4, 542) (5, 496) (6, 465) (7, 446) (8, 437) (9, 427) (10, 416) (11, 410) (12, 408) (13, 409) (14, 404) (15, 399) (16, 395)};
            \legend{No Prov., Prov.}
        \end{axis}
    \end{tikzpicture}
    \caption{Total evaluation runtime of \souffle on all DaCapo benchmarks with each \DOOP (context-insensitive and 1-obj, 1-heap) analysis with and without provenance}
    \label{fig:runtime-threadtest}
\end{subfigure}\hfill
\begin{subfigure}[b]{0.47\textwidth}
    \center
    \begin{tikzpicture}
        \tikzstyle{every node}=[font=\scriptsize],
        \begin{axis}[
                ybar,
            width=\textwidth,
            height=0.7\textwidth,
            bar width=2pt,
            enlargelimits=0.0,
            clip=false,
            enlarge x limits=0.1,
            legend style={
              anchor=north,legend columns=-1},
            ylabel={Memory (MB)},
            xlabel={Threads},
            symbolic x coords={1, 2, 3, 4, 5, 6, 7, 8, 9, 10, 11, 12, 13, 14, 15, 16},
            xtick=data,
            ymax=1500,
            ymin=0,
            visualization depends on=rawy\as\rawy, 
            after end axis/.code={ 
                \draw [ultra thick, white, decoration={snake, amplitude=1pt}, decorate] (rel axis cs:0,1.05) -- (rel axis cs:1,1.05);
            },
        ]
            \addplot coordinates {(1, 982) (2, 988) (3, 991) (4, 993) (5, 995) (6, 997) (7, 999) (8, 1000) (9, 1002) (10, 1003) (11, 1004) (12, 1004) (13, 1005) (14, 1005) (15, 1006) (16, 1006)};
            \addplot coordinates {(1, 1425) (2, 1440) (3, 1449) (4, 1454) (5, 1458) (6, 1462) (7, 1465) (8, 1467) (9, 1469) (10, 1471) (11, 1473) (12, 1475) (13, 1476) (14, 1478) (15, 1479) (16, 1479)};
            \legend{No Prov., Prov.}
        \end{axis}
    \end{tikzpicture}
    \caption{Average evaluation memory usage of \souffle on all DaCapo benchmarks with each \DOOP (context-insensitive and 1-obj, 1-heap) analysis with and without provenance}
    \label{fig:memory-threadtest}
\end{subfigure}
\end{figure*}
Figures~\ref{fig:runtime-threadtest} and \ref{fig:memory-threadtest} show the total runtime and average memory usage for each of the \DOOP DaCapo benchmarks with both \DOOP (context-insensitive and 1-obj-1-heap) analyses, running with multiple threads. The figure demonstrates that the provenance evaluation strategy is scalable, in that the overhead is sustainable with an increasing number of threads. We observe that the overall runtime decreases for provenance and without provenance until 5 threads, and increased thereafter. This is caused by the synchronization of \souffle's rule evaluation system and is not specific to provenance. It is interesting to note that the runtime overhead is larger with fewer threads, being 1.45$\times$ for 1 thread while being 1.23$\times$ for 16 threads. Again, this is related to the underlying hardware architecture providing caches and memory lanes for each core. With more threads, the memory bandwidth to access the logical relations with proof annotations improves. 

The memory usage of the provenance evaluation strategy has a consistent overhead of 1.45$\times$, which aligns with our expectations that there would be a reasonable overhead associated with storing the provenance annotations per tuple. Note that this overhead is constant over any number of threads since the amount of extra information stored overall does not change with the number of threads. 

\textbf{Comparison with current approaches:} The current state of the art in tracking Datalog provenance is to instrument the specification with a given provenance query. The instrumented Datalog specification can then be evaluated using any Datalog engine. One example of this approach is the top-$k$ approach \cite{sp15,Deutch2018}, where Datalog specifications are instrumented based on a provenance query taking the form of a \emph{derivation tree pattern}.

For our experiments, we implemented the instrumentation algorithm presented in \cite{Deutch2018}, and evaluated the resulting Datalog using \souffle, again using \DOOP as the test Datalog specification. Since the instrumentation is done for a specific query, we chose an arbitrary tuple from the $VarPointsTo$ relation in \DOOP. 

\begin{table}[t]
    \centering
\footnotesize
    \begin{tabular}{| l | r | r |}
        \hline
        & Runtime (sec) & Average Memory (MB) \\
        \hline
        Our approach & 3:02.7 & 1,289 \\
        Top-$k$ & 3:02.3 & 1,289 \\
        \hline
    \end{tabular}
    \caption{Runtime and memory usage overheads for our provenance approach compared to \cite{Deutch2018}. Runtime is the total over all context-insensitive DaCapo analyses, and memory usage is the average over all context-insensitive DaCapo analyses.}
    \label{table:topkcomparison}
\end{table}

The results in Figure~\ref{table:topkcomparison} showed that during evaluation time, the 
difference in both runtime and memory usage is less than 1\%, demonstrating that our 
provenance encoding scheme is as scalable as the state-of-the-art. However, with our 
approach, we are able to answer \emph{any} provenance query during proof construction time.

\subsection{Proof Tree Construction}
For the construction of proof trees, the performance of the provenance queries is instrumental. A debugging query 
constitutes a backward search for a rule (i.e., reverting the computational direction of a rule). 
The construction of the proof tree is performed level by level. The expansion of a node in the proof tree
 represents a single debugging query. 

In Figure~\ref{prooftime}, we show the time for proof tree fragments with levels up to a height of 20. 
We initiate the proof tree construction for randomly sampled output tuples in the \DOOP DaCapo benchmarks.
In the figure, we plot the runtime against the number of nodes in the proof tree fragment. Even for 20 levels, these proof trees contain over 15,000 nodes. Considering that full proof trees may have heights over 200, the corresponding full proof trees would be intractable to compute and understand due to exponential growth.
However, this experiment shows that the construction of proof trees is approximately linear in the size of the tree. Therefore, provenance queries can be efficiently computed, and the method will scale well for interactive use. 
The interactive exploration of proof trees is scalable, with each debugging query on average taking less than 1 ms per node. 

\subsection{Characteristics of Proof Trees}
In the following experiments,
we demonstrate the difficulty of proof tree construction for Datalog specifications at large 
scale. Figure~\ref{proofheight} shows the distribution of heights of full proof trees for the DaCapo benchmarks. 
The proof tree  heights can be more than 300. While this may not seem prohibitive, the expected 
number of nodes in the proof tree is exponential in height. A non-linear regression performed on the sizes of actual proof trees, suggests that the branching factor of proof trees is approximately $1.466$ for the DaCapo benchmarks. Therefore, since larger proof trees will have an exponential number of nodes, it is computationally intractable to construct full proof trees for these large specifications. Besides, there is a usability challenge in generating meaningful explanations for the existence of a tuple, which is addressed by the interactive exploration of fragments of a proof tree, with the user exploring only relevant fragments. This is in contrast to a full proof tree, where a user may have to interpret millions of nodes to find an explanation.

\begin{figure}
\captionsetup[subfigure]{justification=centering}
\begin{subfigure}[b]{0.47\textwidth}
    \begin{tikzpicture}\tikzstyle{every node}=[font=\scriptsize]

    \definecolor{color0}{rgb}{0.12156862745098,0.466666666666667,0.705882352941177}

    \begin{axis}[
        width=\textwidth,
        height=0.7\textwidth,
    xlabel={Number of nodes},
    ylabel={Time (seconds)},
    xmin=0, xmax=17860.4533266129,
    ymin=0, ymax=18.4649669857143,
    tick align=outside,
    tick pos=left,
    x grid style={white!69.01960784313725!black},
    y grid style={white!69.01960784313725!black}
    ]
    \addplot [only marks, draw=color0, fill=color0, colormap/viridis]
    table{%
    x                      y
    +1.496400000000000e+04 +1.306290000000000e+01
    +1.531800000000000e+04 +1.345250000000000e+01
    +5.991000000000000e+03 +6.554370000000000e+00
    +5.451000000000000e+03 +5.616100000000000e+00
    +5.464000000000000e+03 +3.705830000000000e+00
    +3.839000000000000e+03 +3.617310000000000e+00
    +5.674000000000000e+03 +5.943610000000000e+00
    +1.358600000000000e+04 +1.297980000000000e+01
    +3.631000000000000e+03 +2.800960000000000e+00
    +1.136900000000000e+04 +1.089360000000000e+01
    +7.185000000000000e+03 +3.755450000000000e+00
    +1.442300000000000e+04 +8.230449999999999e+00
    +7.013000000000000e+03 +4.598570000000000e+00
    +8.412000000000000e+03 +5.383030000000000e+00
    +2.158000000000000e+03 +1.238990000000000e+00
    +6.219000000000000e+03 +3.754490000000000e+00
    +8.720000000000000e+03 +5.998910000000000e+00
    +1.084900000000000e+04 +6.339260000000000e+00
    +4.759000000000000e+03 +3.030670000000000e+00
    +1.122400000000000e+04 +7.664280000000000e+00
    +1.669400000000000e+04 +1.564740000000000e+01
    +9.337000000000000e+03 +8.324759999999999e+00
    +1.571600000000000e+04 +1.552220000000000e+01
    +1.099200000000000e+04 +1.149270000000000e+01
    +5.260000000000000e+02 +2.324460000000000e-01
    +7.235000000000000e+03 +6.555820000000000e+00
    +8.071000000000000e+03 +8.850700000000000e+00
    +1.206800000000000e+04 +1.134560000000000e+01
    +1.219800000000000e+04 +1.246870000000000e+01
    +1.145800000000000e+04 +1.129550000000000e+01
    +1.226000000000000e+04 +7.989440000000000e+00
    +1.607200000000000e+04 +1.187880000000000e+01
    +3.748000000000000e+03 +2.656230000000000e+00
    +4.262000000000000e+03 +2.840810000000000e+00
    +4.444000000000000e+03 +3.769740000000000e+00
    +7.949000000000000e+03 +4.160000000000000e+00
    +7.332000000000000e+03 +5.265710000000000e+00
    +6.182000000000000e+03 +5.053450000000000e+00
    +1.546100000000000e+04 +1.036740000000000e+01
    +1.122600000000000e+04 +8.057200000000000e+00
    +1.132600000000000e+04 +1.155110000000000e+01
    +1.139400000000000e+04 +1.158080000000000e+01
    +1.142500000000000e+04 +1.168040000000000e+01
    +5.246000000000000e+03 +4.959010000000000e+00
    +4.349000000000000e+03 +4.363260000000000e+00
    +6.628000000000000e+03 +5.569760000000000e+00
    +4.991000000000000e+03 +5.260560000000000e+00
    +3.159000000000000e+03 +4.146190000000000e+00
    +1.118200000000000e+04 +1.097340000000000e+01
    +1.170000000000000e+04 +1.108890000000000e+01
    +1.274400000000000e+04 +1.093060000000000e+01
    +1.365200000000000e+04 +1.314270000000000e+01
    +7.635000000000000e+03 +7.776100000000000e+00
    +1.147100000000000e+04 +1.318390000000000e+01
    +8.260000000000000e+03 +8.513700000000000e+00
    +5.005000000000000e+03 +4.282570000000000e+00
    +7.785000000000000e+03 +9.578900000000001e+00
    +6.501000000000000e+03 +8.557869999999999e+00
    +9.325000000000000e+03 +8.120210000000000e+00
    +1.122300000000000e+04 +1.128100000000000e+01
    +1.196300000000000e+04 +9.155900000000001e+00
    +1.382900000000000e+04 +1.068970000000000e+01
    +1.038100000000000e+04 +9.656040000000001e+00
    +4.679000000000000e+03 +3.602790000000000e+00
    +8.619000000000000e+03 +4.679610000000000e+00
    +6.177000000000000e+03 +4.734700000000000e+00
    +7.972000000000000e+03 +6.195850000000000e+00
    +1.125500000000000e+04 +7.362590000000000e+00
    +4.012000000000000e+03 +2.700840000000000e+00
    +1.130800000000000e+04 +8.157230000000000e+00
    +8.362000000000000e+03 +4.985090000000000e+00
    +1.111700000000000e+04 +8.092599999999999e+00
    +1.159600000000000e+04 +7.192260000000000e+00
    +9.751000000000000e+03 +6.553330000000000e+00
    +1.921000000000000e+03 +5.796680000000000e-01
    +1.252400000000000e+04 +8.075350000000000e+00
    +7.949000000000000e+03 +5.544160000000000e+00
    +4.141000000000000e+03 +2.945290000000000e+00
    +1.281600000000000e+04 +8.371040000000001e+00
    +1.129800000000000e+04 +7.939090000000000e+00
    +1.220200000000000e+04 +8.099380000000000e+00
    +7.650000000000000e+03 +6.224580000000000e+00
    +6.781000000000000e+03 +3.925980000000000e+00
    +7.847000000000000e+03 +4.978110000000000e+00
    +3.998000000000000e+03 +2.080270000000000e+00
    +1.177400000000000e+04 +8.991740000000000e+00
    +9.019000000000000e+03 +6.630830000000000e+00
    +1.099500000000000e+04 +6.453680000000000e+00
    +1.618900000000000e+04 +1.042230000000000e+01
    +1.122300000000000e+04 +7.411490000000000e+00
    +1.035700000000000e+04 +9.417920000000001e+00
    +6.983000000000000e+03 +6.239100000000000e+00
    +9.990000000000000e+02 +5.759860000000000e-01
    +2.394000000000000e+03 +1.562900000000000e+00
    +1.784000000000000e+03 +2.214430000000000e+00
    +8.019000000000000e+03 +6.668010000000000e+00
    +3.183000000000000e+03 +2.860000000000000e+00
    +4.895000000000000e+03 +4.411710000000000e+00
    +8.400000000000000e+03 +8.548439999999999e+00
    +1.206300000000000e+04 +1.052140000000000e+01
    +1.228000000000000e+04 +1.196080000000000e+01
    +8.580000000000000e+03 +7.449100000000000e+00
    +7.590000000000000e+02 +7.213810000000000e-01
    +9.188000000000000e+03 +9.852230000000000e+00
    +1.703500000000000e+04 +1.730880000000000e+01
    +1.292700000000000e+04 +1.455800000000000e+01
    +1.495200000000000e+04 +1.759250000000000e+01
    +1.190600000000000e+04 +1.075990000000000e+01
    +4.802000000000000e+03 +5.510990000000000e+00
    +1.169400000000000e+04 +1.092390000000000e+01
    };
    \addplot [color0, forget plot]
    table {%
    14964 12.7191880661533
    15318 13.0250755934765
    5991 4.96571692866482
    5451 4.49910883613793
    5464 4.51034199392098
    3839 3.10619727103913
    5674 4.69180069657033
    13586 11.5284733411495
    3631 2.92646674651026
    11369 9.61278789460857
    7185 5.99743926658539
    14423 12.2517158845662
    7013 5.84881594822497
    8412 7.05767654349372
    2158 1.65366356078412
    6219 5.1627292343984
    8720 7.32381597404609
    10849 9.16346158328638
    4759 3.90115920645532
    11224 9.48749498087449
    16694 14.2140621403599
    9337 7.85695892421108
    15716 13.36898303945
    10992 9.28702631889998
    526 0.243470214480628
    7235 6.04064371959714
    8071 6.76302217395359
    12068 10.2167861477128
    12198 10.3291177255434
    11458 9.68969182096948
    12260 10.3826912472779
    16072 13.6765987448937
    3748 3.02756516655775
    4262 3.47170694351853
    4444 3.6289711524813
    7949 6.65760330860492
    7332 6.12446035843993
    6182 5.1307579391697
    15461 13.1486403290901
    11226 9.48922315899496
    11326 9.57563206501846
    11394 9.63439012111444
    11425 9.66117688198173
    5246 4.32197057878976
    4349 3.54688269175898
    6628 5.51614166003451
    4991 4.10162786842984
    3159 2.51861671007934
    11182 9.45120324034463
    11700 9.89880137354635
    12744 10.8009103524317
    13652 11.585503219125
    7635 6.38627934369114
    11471 9.70092497875254
    8260 6.926335006338
    5005 4.11372511527312
    7785 6.51589270272638
    6501 5.40640234938466
    9325 7.84658985548826
    11223 9.48663089181426
    11963 10.1260567963881
    13829 11.7384469827866
    10381 8.7590679030964
    4679 3.83203208163652
    8619 7.23654297896236
    6177 5.12643749386853
    7972 6.67747735699032
    11255 9.51428174174178
    4012 3.25568467845979
    11308 9.56007846193423
    8362 7.01447209048197
    11117 9.39503745142935
    11596 9.80893611128191
    9751 8.21469179514836
    1921 1.44887445350843
    12524 10.61081075918
    7949 6.65760330860492
    4141 3.3671521672301
    12816 10.8631247647686
    11298 9.55143757133188
    12202 10.3325740817843
    7650 6.39924067959466
    6781 5.64834728625046
    7847 6.56946622446095
    3998 3.2435874316165
    11774 9.96274396400374
    9019 7.58217860305635
    10995 9.28961858608068
    16189 13.7776971649412
    11223 9.48663089181426
    10357 8.73832976565076
    6983 5.82289327641792
    999 0.652184339971776
    2394 1.85758857899958
    1784 1.33049425225624
    8019 6.71808954282137
    3183 2.53935484752498
    4895 4.01867531864728
    8400 7.0473074747709
    12063 10.2124657024116
    12280 10.3999730284826
    8580 7.2028435056132
    759 0.444802965515379
    9188 7.72820965423607
    17035 14.5087165099
    12927 10.9590386504547
    14952 12.7088189974305
    11906 10.0768037199548
    4802 3.93831503604542
    11694 9.89361683918494
    };
    \end{axis}

    \end{tikzpicture}
\caption{Proof Tree Construction Time}
\label{prooftime}
\end{subfigure}\hfill
\begin{subfigure}[b]{0.47\textwidth}
    \begin{tikzpicture}\tikzstyle{every node}=[font=\scriptsize]

    \definecolor{color0}{rgb}{0.12156862745098,0.466666666666667,0.705882352941177}

    \begin{axis}[
        width=\textwidth,
        height=0.7\textwidth,
    xlabel={Height},
    ylabel={Number},
    xmin=-11.15, xmax=300.15,
    ymin=0, ymax=653776.2,
        ybar=0pt,
    tick align=outside,
    tick pos=left,
    x grid style={white!69.01960784313725!black},
    y grid style={white!69.01960784313725!black}
    ]
    \draw[fill=color0,draw opacity=0] (axis cs:3,0) rectangle (axis cs:5.83,374);
    \draw[fill=color0,draw opacity=0] (axis cs:5.83,0) rectangle (axis cs:8.66,2068);
    \draw[fill=color0,draw opacity=0] (axis cs:8.66,0) rectangle (axis cs:11.49,4279);
    \draw[fill=color0,draw opacity=0] (axis cs:11.49,0) rectangle (axis cs:14.32,6809);
    \draw[fill=color0,draw opacity=0] (axis cs:14.32,0) rectangle (axis cs:17.15,11638);
    \draw[fill=color0,draw opacity=0] (axis cs:17.15,0) rectangle (axis cs:19.98,9867);
    \draw[fill=color0,draw opacity=0] (axis cs:19.98,0) rectangle (axis cs:22.81,20779);
    \draw[fill=color0,draw opacity=0] (axis cs:22.81,0) rectangle (axis cs:25.64,23617);
    \draw[fill=color0,draw opacity=0] (axis cs:25.64,0) rectangle (axis cs:28.47,25025);
    \draw[fill=color0,draw opacity=0] (axis cs:28.47,0) rectangle (axis cs:31.3,31669);
    \draw[fill=color0,draw opacity=0] (axis cs:31.3,0) rectangle (axis cs:34.13,34166);
    \draw[fill=color0,draw opacity=0] (axis cs:34.13,0) rectangle (axis cs:36.96,23210);
    \draw[fill=color0,draw opacity=0] (axis cs:36.96,0) rectangle (axis cs:39.79,44418);
    \draw[fill=color0,draw opacity=0] (axis cs:39.79,0) rectangle (axis cs:42.62,75955);
    \draw[fill=color0,draw opacity=0] (axis cs:42.62,0) rectangle (axis cs:45.45,115104);
    \draw[fill=color0,draw opacity=0] (axis cs:45.45,0) rectangle (axis cs:48.28,134475);
    \draw[fill=color0,draw opacity=0] (axis cs:48.28,0) rectangle (axis cs:51.11,118184);
    \draw[fill=color0,draw opacity=0] (axis cs:51.11,0) rectangle (axis cs:53.94,63932);
    \draw[fill=color0,draw opacity=0] (axis cs:53.94,0) rectangle (axis cs:56.77,82566);
    \draw[fill=color0,draw opacity=0] (axis cs:56.77,0) rectangle (axis cs:59.6,60225);
    \draw[fill=color0,draw opacity=0] (axis cs:59.6,0) rectangle (axis cs:62.43,63657);
    \draw[fill=color0,draw opacity=0] (axis cs:62.43,0) rectangle (axis cs:65.26,55957);
    \draw[fill=color0,draw opacity=0] (axis cs:65.26,0) rectangle (axis cs:68.09,67243);
    \draw[fill=color0,draw opacity=0] (axis cs:68.09,0) rectangle (axis cs:70.92,36861);
    \draw[fill=color0,draw opacity=0] (axis cs:70.92,0) rectangle (axis cs:73.75,59323);
    \draw[fill=color0,draw opacity=0] (axis cs:73.75,0) rectangle (axis cs:76.58,64988);
    \draw[fill=color0,draw opacity=0] (axis cs:76.58,0) rectangle (axis cs:79.41,79321);
    \draw[fill=color0,draw opacity=0] (axis cs:79.41,0) rectangle (axis cs:82.24,69982);
    \draw[fill=color0,draw opacity=0] (axis cs:82.24,0) rectangle (axis cs:85.07,61325);
    \draw[fill=color0,draw opacity=0] (axis cs:85.07,0) rectangle (axis cs:87.9,36036);
    \draw[fill=color0,draw opacity=0] (axis cs:87.9,0) rectangle (axis cs:90.73,68596);
    \draw[fill=color0,draw opacity=0] (axis cs:90.73,0) rectangle (axis cs:93.56,94985);
    \draw[fill=color0,draw opacity=0] (axis cs:93.56,0) rectangle (axis cs:96.39,115357);
    \draw[fill=color0,draw opacity=0] (axis cs:96.39,0) rectangle (axis cs:99.22,107580);
    \draw[fill=color0,draw opacity=0] (axis cs:99.22,0) rectangle (axis cs:102.05,117744);
    \draw[fill=color0,draw opacity=0] (axis cs:102.05,0) rectangle (axis cs:104.88,69718);
    \draw[fill=color0,draw opacity=0] (axis cs:104.88,0) rectangle (axis cs:107.71,98978);
    \draw[fill=color0,draw opacity=0] (axis cs:107.71,0) rectangle (axis cs:110.54,115423);
    \draw[fill=color0,draw opacity=0] (axis cs:110.54,0) rectangle (axis cs:113.37,140316);
    \draw[fill=color0,draw opacity=0] (axis cs:113.37,0) rectangle (axis cs:116.2,146366);
    \draw[fill=color0,draw opacity=0] (axis cs:116.2,0) rectangle (axis cs:119.03,149006);
    \draw[fill=color0,draw opacity=0] (axis cs:119.03,0) rectangle (axis cs:121.86,108361);
    \draw[fill=color0,draw opacity=0] (axis cs:121.86,0) rectangle (axis cs:124.69,181005);
    \draw[fill=color0,draw opacity=0] (axis cs:124.69,0) rectangle (axis cs:127.52,208362);
    \draw[fill=color0,draw opacity=0] (axis cs:127.52,0) rectangle (axis cs:130.35,216040);
    \draw[fill=color0,draw opacity=0] (axis cs:130.35,0) rectangle (axis cs:133.18,193457);
    \draw[fill=color0,draw opacity=0] (axis cs:133.18,0) rectangle (axis cs:136.01,271216);
    \draw[fill=color0,draw opacity=0] (axis cs:136.01,0) rectangle (axis cs:138.84,123948);
    \draw[fill=color0,draw opacity=0] (axis cs:138.84,0) rectangle (axis cs:141.67,192027);
    \draw[fill=color0,draw opacity=0] (axis cs:141.67,0) rectangle (axis cs:144.5,219824);
    \draw[fill=color0,draw opacity=0] (axis cs:144.5,0) rectangle (axis cs:147.33,158114);
    \draw[fill=color0,draw opacity=0] (axis cs:147.33,0) rectangle (axis cs:150.16,142076);
    \draw[fill=color0,draw opacity=0] (axis cs:150.16,0) rectangle (axis cs:152.99,134783);
    \draw[fill=color0,draw opacity=0] (axis cs:152.99,0) rectangle (axis cs:155.82,187418);
    \draw[fill=color0,draw opacity=0] (axis cs:155.82,0) rectangle (axis cs:158.65,155287);
    \draw[fill=color0,draw opacity=0] (axis cs:158.65,0) rectangle (axis cs:161.48,188012);
    \draw[fill=color0,draw opacity=0] (axis cs:161.48,0) rectangle (axis cs:164.31,333949);
    \draw[fill=color0,draw opacity=0] (axis cs:164.31,0) rectangle (axis cs:167.14,491975);
    \draw[fill=color0,draw opacity=0] (axis cs:167.14,0) rectangle (axis cs:169.97,212762);
    \draw[fill=color0,draw opacity=0] (axis cs:169.97,0) rectangle (axis cs:172.8,495440);
    \draw[fill=color0,draw opacity=0] (axis cs:172.8,0) rectangle (axis cs:175.63,622644);
    \draw[fill=color0,draw opacity=0] (axis cs:175.63,0) rectangle (axis cs:178.46,445742);
    \draw[fill=color0,draw opacity=0] (axis cs:178.46,0) rectangle (axis cs:181.29,346731);
    \draw[fill=color0,draw opacity=0] (axis cs:181.29,0) rectangle (axis cs:184.12,570383);
    \draw[fill=color0,draw opacity=0] (axis cs:184.12,0) rectangle (axis cs:186.95,123090);
    \draw[fill=color0,draw opacity=0] (axis cs:186.95,0) rectangle (axis cs:189.78,120582);
    \draw[fill=color0,draw opacity=0] (axis cs:189.78,0) rectangle (axis cs:192.61,116842);
    \draw[fill=color0,draw opacity=0] (axis cs:192.61,0) rectangle (axis cs:195.44,77165);
    \draw[fill=color0,draw opacity=0] (axis cs:195.44,0) rectangle (axis cs:198.27,87076);
    \draw[fill=color0,draw opacity=0] (axis cs:198.27,0) rectangle (axis cs:201.1,80223);
    \draw[fill=color0,draw opacity=0] (axis cs:201.1,0) rectangle (axis cs:203.93,53779);
    \draw[fill=color0,draw opacity=0] (axis cs:203.93,0) rectangle (axis cs:206.76,64922);
    \draw[fill=color0,draw opacity=0] (axis cs:206.76,0) rectangle (axis cs:209.59,106810);
    \draw[fill=color0,draw opacity=0] (axis cs:209.59,0) rectangle (axis cs:212.42,93632);
    \draw[fill=color0,draw opacity=0] (axis cs:212.42,0) rectangle (axis cs:215.25,119174);
    \draw[fill=color0,draw opacity=0] (axis cs:215.25,0) rectangle (axis cs:218.08,71159);
    \draw[fill=color0,draw opacity=0] (axis cs:218.08,0) rectangle (axis cs:220.91,45496);
    \draw[fill=color0,draw opacity=0] (axis cs:220.91,0) rectangle (axis cs:223.74,60676);
    \draw[fill=color0,draw opacity=0] (axis cs:223.74,0) rectangle (axis cs:226.57,71115);
    \draw[fill=color0,draw opacity=0] (axis cs:226.57,0) rectangle (axis cs:229.4,46189);
    \draw[fill=color0,draw opacity=0] (axis cs:229.4,0) rectangle (axis cs:232.23,19734);
    \draw[fill=color0,draw opacity=0] (axis cs:232.23,0) rectangle (axis cs:235.06,86317);
    \draw[fill=color0,draw opacity=0] (axis cs:235.06,0) rectangle (axis cs:237.89,16357);
    \draw[fill=color0,draw opacity=0] (axis cs:237.89,0) rectangle (axis cs:240.72,16764);
    \draw[fill=color0,draw opacity=0] (axis cs:240.72,0) rectangle (axis cs:243.55,19041);
    \draw[fill=color0,draw opacity=0] (axis cs:243.55,0) rectangle (axis cs:246.38,14784);
    \draw[fill=color0,draw opacity=0] (axis cs:246.38,0) rectangle (axis cs:249.21,10989);
    \draw[fill=color0,draw opacity=0] (axis cs:249.21,0) rectangle (axis cs:252.04,7678);
    \draw[fill=color0,draw opacity=0] (axis cs:252.04,0) rectangle (axis cs:254.87,3091);
    \draw[fill=color0,draw opacity=0] (axis cs:254.87,0) rectangle (axis cs:257.7,10076);
    \draw[fill=color0,draw opacity=0] (axis cs:257.7,0) rectangle (axis cs:260.53,5104);
    \draw[fill=color0,draw opacity=0] (axis cs:260.53,0) rectangle (axis cs:263.36,5533);
    \draw[fill=color0,draw opacity=0] (axis cs:263.36,0) rectangle (axis cs:266.19,4191);
    \draw[fill=color0,draw opacity=0] (axis cs:266.19,0) rectangle (axis cs:269.02,3124);
    \draw[fill=color0,draw opacity=0] (axis cs:269.02,0) rectangle (axis cs:271.85,1738);
    \draw[fill=color0,draw opacity=0] (axis cs:271.85,0) rectangle (axis cs:274.68,3212);
    \draw[fill=color0,draw opacity=0] (axis cs:274.68,0) rectangle (axis cs:277.51,1078);
    \draw[fill=color0,draw opacity=0] (axis cs:277.51,0) rectangle (axis cs:280.34,539);
    \draw[fill=color0,draw opacity=0] (axis cs:280.34,0) rectangle (axis cs:283.17,2178);
    \draw[fill=color0,draw opacity=0] (axis cs:283.17,0) rectangle (axis cs:286,330);
    \end{axis}

    \end{tikzpicture}
    \caption{Proof Tree Heights}
    \label{proofheight}
\end{subfigure}
\caption{Proof Tree Construction and Statistics}
\end{figure}
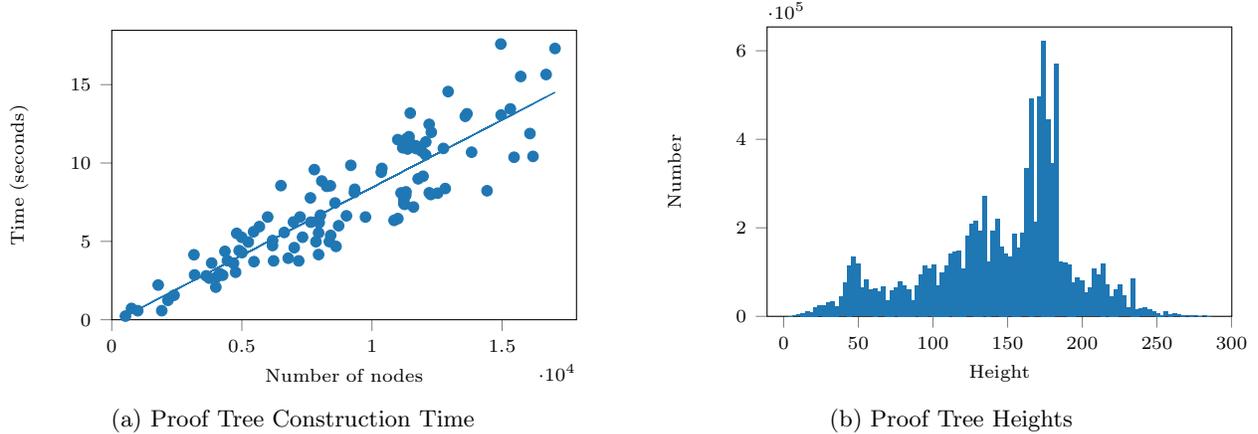

\section{Related Work} \label{sec:related}
Debugging for logic programming languages has a long history, with work having been done on algorithmic debugging strategies since the 1980s \cite{ad89,ap83}. These works present a framework for the algorithmic debugging method for Prolog programs, where a system asks the user questions about the intended model of the program to find buggy rules. However, they are based on the SLDNF resolution of Prolog which is not truly declarative, and thus the semantics of Datalog differ. Our method aligns practically with the interactive debugging frameworks presented here, but applied to the bottom-up evaluation of Datalog, and with sophisticated and efficient techniques to generate the debugging information.

\textbf{Debugging and Provenance.} Our method also fits into the established frameworks for provenance in Datalog~\cite{pd09} and debugging for Datalog~\cite{tf08}. The proof trees generated by our method are analogous to the computation graphs presented in~\cite{tf08}, and are equally effective for debugging. They can also be seen as an extension of \emph{how}-provenance~\cite{ww01}. However, our hybrid method for generating proof trees is novel, i.e., it permits multiple debugging queries in a single debugging cycle. Hence, our method is especially useful for debugging large Datalog specifications.

\textbf{Provenance in Datalog.} Methods for computing provenance in Datalog has been a well-explored field \cite{ep93,dd12,cd14,sp15,dw15,ec17}, however with the caveat that all these previous approaches store the full provenance information during the evaluation of Datalog. \cite{dd12}, for example, stores the whole computation graph as an auxiliary relation during Datalog evaluation, which may be many times larger than the IDB itself in large analysis use cases. Approaches such as \cite{sp15,ec17} attempt to reduce the impact of provenance storage by only storing information relevant for a particular provenance query, which is given before the instrumentation and evaluation of Datalog. Thus, in these approaches, the Datalog specification needs to be re-evaluated for each different provenance query and therefore extends the investigation phase of the debugging cycle. The closest approach to ours is perhaps \cite{cd14}, where a boolean circuit representation for provenance is described, as well as an algorithm for generating such a boolean circuit during Datalog evaluation. However, there is no mechanism for exploring an understandable provenance representation, and no practical implementation of this work. Therefore, our approach is novel by minimizing the storage overhead of provenance information, allowing interactive exploration of proof trees, all while providing an effective integration into an existing semi-na\"ive Datalog engine.

\textbf{Other Applications for Datalog Provenance.} Debugging Datalog specifications is not the only use case for provenance, with user-guided approaches~\cite{ei17,ug15,ug18,ar14} for program analysis also relying on tracking the origins of data. In \cite{ug18,ei17,ug15}, a user may tag certain static analysis alarms, to increase or decrease their importance in the next analysis cycle. In \cite{ar14}, the analysis system automatically generates an appropriate abstraction, by iteratively trying and refining failing abstractions. All these approaches rely on an annotation framework for Datalog: the user-guided systems require the user to add an annotation representing the importance of an alarm, and the abstraction refinement system requires the system to tag failing analyses with annotations. In any case, our provenance evaluation strategy would fit well into these systems, by providing an annotation framework at the Datalog engine level.

\textbf{Provenance in Databases.} Outside of Datalog, provenance has also been a 
focus of the database community, being useful for understanding the origins of data 
in large database systems. For instance, \emph{Trio}~\cite{ta05,ai06} and 
\emph{Perm}~\cite{pp09,us13} are two such systems implementing provenance systems for 
relational databases. These systems focus on tracking the \emph{lineage} of data in 
a database, rather than on debugging a query, and thus it is essential that the 
provenance information is stored directly alongside the data. It is important to note 
that in the context of databases, Datalog acts as a powerful query language rather 
then a specification logic. As a result, the Datalog specifications in these use cases 
typically consist of fewer rules~\cite{openrule}, that do not exhibit 
complex patterns found in program analysis 
benchmarks~\cite{doop09}. Further study has also been undertaken in querying
database provenance, with \cite{dp15, ur18} both presenting mechanisms to construct
provenance information lazily after the database query is run. Thus, similarly to our
approach, these allow the querying of arbitrary parts of provenance information.
However, both approaches are applied to database systems, with \cite{dp15}
reconstructing provenance information based on tracking file I/O and system calls during
query evaluation, and \cite{ur18} based on system logs produced by the database system.
Therefore, with the highly complex recursive nature of real-world Datalog specifications, similar
information may blow up drastically, and experiments of both works only show scalability
up to 10,000 tuples. Moreover, these approaches are unsuitable for our setting of
in-memory analysis workloads.

\section{Conclusion} \label{sec:conclusion}
In this paper we have presented a novel provenance evaluation strategy for 
Datalog specifications. The provenance evaluation strategy extends
tuples in the IDB with proof annotations. With the help of proof annotations, 
provenance queries can construct minimal proof trees incrementally. 
Our method has very small overheads at
logic evaluation time in comparison to standard top-down evaluation or a na\"ive provenance methods
that encode provenance information explicitly. 
Hence, our method enables debugging of large-scale logic specifications for the first time.
We have implemented our provenance method in a high performance Datalog engine called
\souffle~\cite{so16}, and demonstrated its feasibility through the \DOOP program analysis framework. 
We show that the runtime overheads of the provenance evaluation strategy are 
approximately 1.27$\times$, and the memory overheads are 1.45$\times$. 

\bibliographystyle{abbrv}
\bibliography{bibliography}

\end{document}